\documentclass[11pt,reqno ]{amsart}
\usepackage{}
\usepackage{mathrsfs}
\usepackage{amsmath,amsfonts,amssymb,amscd,amsthm,bbm}
\usepackage{graphicx}
\usepackage{color}
\usepackage{ifpdf}
\usepackage{cite}
\newcommand{\ua}{u^\alpha}
\newcommand{\ub}{u^\beta}

\newcommand{\Tab}{T^{\alpha\beta}}

\title[The Riemann problem of relativistic Euler system with Synge energy]{The Riemann problem of  relativistic \\ Euler system    with Synge energy}

\author[Ruggeri]{Tommaso Ruggeri}
\address[T.  Ruggeri]{\newline Department of Mathematics and Alma Mater Research Center on Applied Mathematics AM$^2$,
University of Bologna, Bologna, Italy}
\email{tommaso.ruggeri@unibo.it}

\author[Xiao]{Qinghua Xiao}
\address[Q.H. Xiao]{\newline  Wuhan Institute of Physics and Mathematics, Chinese Academy of Sciences, Wuhan 430071, China}
\email{xiaoqh@wipm.ac.cn}

\author[Zhao]{HuiJiang Zhao}
\address[H.-J. Zhao]{\newline School of Mathematics and Statistics, Wuhan University, Wuhan 430072, China}
\email{hhjjzhao@hotmail.com}

\newtheorem{theorem}{Theorem}[section]
\newtheorem{lemma}{Lemma}[section]
\newtheorem{corollary}{Corollary}[section]
\newtheorem{proposition}{Proposition}[section]
\newtheorem{remark}{Remark}[section]
\newtheorem{conjecture}{Conjecture}[section]

\newcommand{\bbr}{\mathbb R}

\numberwithin{equation}{section}

\begin{document}
%%%%%%%%%%%%%%%%

\date{\today}

%\subjclass{    } \keywords{Relativistic Landau-Maxwell system, Vlasov-Poisson-Landau system, uniform $L^2$-stability, equilibrium}
%\thanks{\textbf{Acknowledgment.} The work of...}

\begin{abstract}
In this paper, we study the Riemann problem of relativistic Euler system for rarefied monatomic  and diatomic gases when the constitutive equation for the energy is the  Synge equation that is the only one compatible with the relativistic kinetic theory. The Synge equation is involved with modified Bessel functions of the second kind and this makes the relativistic
Euler system quite complex. Based on delicate estimates of the modified Bessel functions of
the second kind, we provide a detailed investigation of basic hyperbolic properties and the structure
of elementary waves, especially for the structure of shock
waves and in this way, the mathematical theory of the Riemann problem for these relativistic Euler system, which
is analogous to the corresponding theory of the classical ones,  is rigorously provided.

\vspace{.3cm}	
% Include keywords, PACS and mathematical
%subject classification numbers as needed ????.
\keywords{Riemann problem, Relativistic Euler Fluid, Synge Energy, Relativistic Kinetic Theory}
%\subclass{76N10,76Y05,35L65}
\end{abstract}
\maketitle
\centerline{\date}
\tableofcontents
%\newpage
%\tableofcontents
\section{Introduction}
\setcounter{equation}{0}One of the main problems in hyperbolic systems is the {\it Riemann problem}.
This problem was proposed by Riemann considering a gas that is initially separated into two regions by a thin diaphragm.
The gases in the two regions are in different equilibrium thermodynamic states, respectively.
The question raised by Riemann is what happens when the diaphragm is put away.
In literature, by extension of this problem, the Riemann problem deals with every solution of a system of conservation laws in one-space dimension along the $x$ axis when the initial data composed of two different constant states $({\bf u}_L, {\bf u}_R)$ are connected with a jump at $x=0$.

The Riemann problem for hyperbolic conservation systems was completely solved mainly by P. Lax~\cite{Lax73}.
It was shown that the solution of the Riemann problem for hyperbolic systems of conservation laws is a combination of the rarefaction waves, contact waves, and shock waves (see e.g. \cite{Dafermos} and references therein).

A huge literature of the Riemann problem exists, in particular, many numerical results have been obtained by using the Riemann solvers (see e.g., \cite{Toro}).

For the classical Euler system, there have been enormous works (see \cite{Bressan-2000,Dafermos,Glimm-CPAM-1965,Smoller-1994} for instance). For brevity, we only list some of them: the global existence, as well as the sharp decay rate, was obtained for the entropy solutions with small amplitude in the celebrated work of Glimm and Lax \cite{Glimm-Lax-MAMS-1970}; the ``large data" global existence theorem for weak solutions was initiated by Nishida \cite{Nishida-PJA-1968}. It is well known that the Boltzmann equation is related to the systems of fluid dynamics for rarefied gas. This fact is revealed in the works such as \cite{Nicolaenko-Thurber-JM-1975,Caflisch-Nicolaenko-CMP-1982} for the shock profile solutions of the classical Boltzmann equation, \cite{Yu-CPAM-2005,Huang-Wang-Yang-CMP-2010,Huang-Wang-Wang-Yang-SIAM-2013} for the hydrodynamic limits from classical Boltzmann equation to Euler system with waves and \cite{Liu-Yang-Yu-Zhao-ARMA-2006,Ukai-Yang-Yu-CMP-2003,Ukai-Yang-Yu-CMP-2004,Xin-Yang-Yu-ARMA-2012} for the nonlinear stability of waves and boundary layers of the classical Boltzmann equation.

\smallskip

The aim of this paper is to consider the problematic of Riemann problem in the relativistic framework.
Let $V^\alpha$ and $T^{\alpha\beta}$ be the particle-particle flux and energy-momentum tensor, respectively \cite{Anile, Cercignani-Kremer, Groot-Leeuwen-Weert-1980,Taub-1967,Synge}:
\begin{align}\label{EuVT}
V^\alpha := \rho \ua, \quad \Tab := p h^{\alpha\beta} +\frac{e}{c^2} \ua \ub.
\end{align}
Then, the field equations for relativistic single fluid are the conservation of particle numbers and energy-momentum tensor in Minkowski space:
\begin{align}\label{massmomenergy}
& \partial_\alpha V^\alpha = 0,  \qquad
\partial_\alpha T^{\alpha\beta}= 0,
\end{align}
where  $\rho = n m$ is the density, $n$ is the particle number, $m$ is the mass in  rest frame, $u^\alpha \equiv (u^0 = \Gamma c, u^i= \Gamma v^i)$  is the four-velocity vector, $\Gamma=1/  \sqrt{1-{v^2}/{c^2}}$ is the Lorentz factor, $v^i$ is the velocity,  $h^{\alpha \beta} =  \ua \ub /{c^2} - g^{\alpha\beta}$ is the projector tensor, $g^{\alpha\beta}$  is the metric tensor with signature $(+,-,-,-)$, $p$ is the pressure,
\begin{equation}
e= \rho(c^2+\varepsilon) \label{energia-e}
\end{equation}
 is the  energy,  that is the sum of internal energy ($\varepsilon$  is the internal energy density) and the energy in the rest frame, $c$ is the light velocity;  $\partial_\alpha = \partial / \partial x_\alpha$, $x^\alpha \equiv (x^0 = c t, x^i)$ are the space-time coordinates and the greek indices run from $0$ to $4$ while the Latin indices from $1$ to $3$ and, as usual, contract indices indicate summation.

For two dimensional space-time case, the system \eqref{massmomenergy} with \eqref{EuVT} is expressed as
\begin{align} \label{main1}
\begin{aligned}
&  \partial_t\left(\frac{\rho c}{\sqrt{c^2-v^2}}\right) +   \partial_x\left(\frac{\rho cv}{\sqrt{c^2-v^2}}\right)=0, \\
&  \partial_t\left(\frac{(e+p)v}{c^2-v^2}\right)+\partial_x \left(\frac{(e+p)v^2}{c^2-v^2}+p\right)=0,\\
&  \partial_t\left(\frac{(e+p)v^2}{c^2-v^2}+e\right)+ \partial_x \left(\frac{(e+p)c^2v}{c^2-v^2}\right)=0.
 \end{aligned}
\end{align}
We need the constitutive equation
\begin{equation}\label{pre}
p \equiv p(\rho,e)
\end{equation}
 to close the system \eqref{main1}. This is usually obtained, in parametric form, through the thermal and caloric equation of state
\begin{equation}\label{clter}
p\equiv p(\rho,T), \qquad e \equiv e(\rho,T),
\end{equation}
where $T$ is the temperature.

To the system \eqref{main1} with \eqref{pre}  we prescribe the   Riemann initial data
\begin{equation}\label{ini-data}
\mathbf{u}_0(x)=\left\{\begin{array}{cc}\mathbf{u}_L=(\rho_L,~v_L,~e_L),\quad &x<0,\\
\mathbf{u}_R=(\rho_R,~v_R,~e_R),\quad &x>0, \end{array}\right.
\end{equation}
where $\mathbf{u}_L$ and $\mathbf{u}_R$
are two different constant states: $\mathbf{u}_L\neq \mathbf{u}_R$.

%%%%%%%%%%%%%%%%%%%%%%%%

In 1948, Taub \cite{Taub-PR-1948} derived the equations (\ref{main1}) for a relativistic fluid and the Rankine-Hugoniot equations of the shock waves assuming as constitutive functions the  pressure  and the internal energy   of polyatomic polytropic classical case:
\begin{equation}\label{state-rarepoly}
p=\frac{k_B}{m} \rho T, \quad \varepsilon=\frac{D}{2}\frac{k_B}{m} T, \quad \rightarrow \quad \varepsilon = \frac{p}{\rho (k-1)}, \quad \rightarrow \quad p=(k-1)(e - \rho c^2),
\end{equation}
where $D = 2/(k-1)$ is related to the degree of freedom and $k=c_p/c_V>1$ is the ratio of specific heats and $k_B$ is the Boltzmann constant.
 Smoller and Temple \cite{Smoller-Temple-CMP-1993} considered  as constitutive  equation  $p=\sigma^2e$ ($\sigma$ is a constant) that substantially corresponds to  the ultra-relativistic regime as we will see later. In this case the authors took into account only the second and third equation of the system \eqref{main1} because the first equation is independent, and established the global existence of entropy solutions to the Cauchy problem with arbitrary initial data of finite total variation (see also Wissman  \cite{Wissman-CMP-2011}). Chen \cite{Chen-CPDE-1995} extended this result to the case of a constitutive equation corresponding to an isentropic classical gas for which $p=\sigma^2\rho^{k}$ and discussed the Riemann problem of the relativistic Euler system (\ref{main1}). The same author in \cite{Chen-ARMA-1997} considered as constitutive equations   \eqref{state-rarepoly} that corresponds to a polyatomic classical gas.
On the other hand, for smooth solutions to  the  ultra-relativistic Euler system in  $(3+1)$-dimensional space-time, Makino-Ukai \cite{Makino-Ukai-JMKU-1995,Makino-Ukai-KMJ-1995} established the local existence of solutions with data
away from vacuum  applying Friedrichs-Lax-Kato's theory, and Lefloch-Ukai \cite{Lefloch-Ukai-KRM-2009} further extended it to the case with vacuum; the singularity formation of smooth solutions was studied by Pan-Smoller \cite{Pan-Smoller-CMP-2006}. For more works about the relativistic Euler system, we refer the interested readers to \cite{Chen-Li-JDE-2004,Chen-Li-ZAMP-2004,Li-Wang-CAM-2006,Ruan-Zhu-NA-2005} and the references therein.

\smallskip

%%%%%%%%%%%%%%%%
The previous constitutive equations in \cite{Chen-CPDE-1995, Chen-ARMA-1997, Makino-Ukai-JMKU-1995, Makino-Ukai-KMJ-1995,Smoller-Temple-CMP-1993} are too much simplified, either only verified in the ultra-relativistic limit or verified in the classical limit. To have more realistic equations in the relativistic regime, at least for rarefied gas, we need to justify this at the mesoscopic scale using the kinetic theory.
If we take into account the relativistic kinetic framework, we have   the Boltzmann-Chernikov equation:
\begin{equation}\label{B-C}
p^\alpha \partial_\alpha f =  Q,
\end{equation}
where $f\equiv f(x^\alpha,p^\beta)$ is the distribution function, $p^\alpha $ is the four-momentum with the property $p^\alpha p_\alpha=m^2 c^2$, and $Q$ is the collisional term. Taking the first $2$-blocks of the tensorial moments, we have:
\begin{align}\label{moments3}
\begin{split}
&  V^\alpha = m c \int_{\Re^3}  f p^\alpha
\, d \vec{P} \,  \, , \quad T^{\alpha \beta} =
c \int_{\Re^3}   f   p^\alpha
p^\beta \,  \, d \vec{P},
\end{split}
\end{align}
with
\begin{equation*}
d \vec{P} =  \frac{dp^1 \, dp^2 \,
    dp^3}{p^0}.
\end{equation*}
 In the case of non-degenerate gases, the constitutive equations \eqref{clter} can be calculated via kinetic theory with the J\"uttner equilibrium distribution function
 \begin{equation*}
 f_J= \frac{n \gamma}{ K_2(\gamma)} \frac{1}{4 \pi m^3
     c^3} e^{- \frac{\gamma}{mc^2}   u_\beta
     p^\beta},
 \end{equation*}
 as follows:
\begin{align}
&   p =  \frac{m nc^2}{\gamma} = \frac{k_B}{m}\rho T  \, ,\label{press} \\
& e= \frac{n m c^2}{  K_2(\gamma)}  \left[ K_3 \left(\gamma
\right) - \frac{1}{\gamma}  K_2
\left(\gamma   \right) \right], \label{caloric}
\end{align}
where
$K_j(\gamma), (j=0, 1, 2, \ldots,)$ are the modified second order Bessel functions, and
 $\gamma$ is a dimensionless variable defined as
\begin{equation}\label{temp}
\gamma=\frac{mc^2}{k_BT}.
\end{equation}
We recall that a fluid can be considered in a relativistic context if $\gamma$ is very small. This means that the bodies are so hot that the mean kinetic energy of particles becomes comparable with their rest energy or even surpasses that energy or the mass is extremely small. Therefore it is
of considerable interest in several areas of astrophysics and nuclear physics. The two limits $\gamma \rightarrow 0$ and $\gamma \rightarrow \infty$ correspond respectively to the ultra-relativistic limit and classical limit.

The expression of energy \eqref{caloric} is called the Synge  energy  \cite{Synge}. In the classical limit ($\gamma\rightarrow\infty $), by taking into account the expansion of the Bessel functions:
\begin{align*}
\begin{split}
& K_3(\gamma)= \sqrt{\frac{\pi}{2}} \gamma^{-1/2} e^{-\gamma} \left[ 1 +
\frac{35}{8} \, \frac{1}{\gamma} + o \left( \frac{1}{\gamma^2} \right) \right] \, , \notag\\
& K_2(\gamma)= \sqrt{\frac{\pi}{2}} \gamma^{-1/2} e^{-\gamma} \left[ 1 +
\frac{15}{8}\, \frac{1}{\gamma} + o \left( \frac{1}{\gamma^2} \right) \right] \, , \notag
\end{split}
\end{align*}
 the Synge  energy converges to
\begin{equation*}
\frac{e}{\rho}=  c^2 +  \varepsilon, \quad \text{with }\quad \varepsilon=\frac{3}{2}\frac{k_B}{m} T.
\end{equation*}
The  expression  of internal energy $\varepsilon$ shows that both classical and relativistic  kinetic theories are valid only for rarefied monatomic gases. In fact, the usual expression in classical theory of internal energy  (for polyatomic polytropic gas)  is \eqref{state-rarepoly},
where $D= 3 + f^i$ is related to the degrees of freedom of a molecule given by the sum of the space dimension $3$ for the translational motion and the contribution from the internal degrees of freedom $f^i \geq 0$ due to the internal motion (rotation and vibration). For monatomic gases, $D=3$.
In the ultra-relativistic limit  $\gamma \rightarrow 0$, the Synge energy equation \eqref{caloric} converges to $e=3p$.

In this context, we mention the  work of Speck and Strain  \cite{Speck-Strain-CMP-2011} where the local existence of smooth solutions to the Relativistic Euler system derived from relativistic Boltzmann equation was presented with the energy currents method introduced by Christodoulou \cite{Christodoulou-2007}.

Recently, a big effort was made to construct a Rational Extended Thermodynamics (RET) theory, in the classical framework, that goes beyond the monatomic gas case. In fact,  Ruggeri and Sugiyama with coworkers gave a series of papers in these years on this subject and the results are summarized in their recent book \cite{book}.
% % % % % % % %
 Pennisi and Ruggeri generalized this idea to the relativistic framework for a gas with internal structure both in the case of dissipative gas \cite{Pennisi_Ruggeri} and the most simple case of Euler fluid \cite{PR2}.
 They started from the classical ideas for polyatomic gases introduced first by Borgnakke and Larsen \cite{Borg}
 and  proposed a generalized Boltzmann-Chernikov equation that has the same form of  \eqref{B-C}   but has  the extended distribution function $f \equiv f(x^\alpha,p^\beta,\mathcal{I})$,   depending on an extra variable $\mathcal{I}$ that takes
  into account the energy due to the internal degrees of freedom of a molecule.
  The authors considered instead of  \eqref{moments3}, the following moments:
 \begin{align}\label{14n}
 \begin{split}
 &   V^\alpha = m c \int_{\Re^3} \int_0^{+\infty} f p^\alpha
 \phi(\mathcal{I}) \, d \vec{P} \, d \, \mathcal{I}, \\
 & T^{\alpha \beta} =
 \frac{1}{mc} \int_{\Re^3} \int_0^{+\infty} f \left( mc^2 + \mathcal{I} \right) p^\alpha
 p^\beta \, \phi(\mathcal{I}) \, d \vec{P} \, d \,
 \mathcal{I}.
 \end{split}
 \end{align}
 The meaning of $\eqref{14n}_{2}$ is that the energy and the momentum in relativity are
 components of the same tensor and we expect that, besides the energy at rest, there is a contribution
 from the degrees of freedom of the  gas  due to  the internal structure, as in the case of a classical polyatomic gas.
 $\phi(\mathcal{I})$ is the state density of the internal mode, that is, $\phi(\mathcal{I}) \, d  \mathcal{I}$ represents the number of the internal states of a molecule having the internal energy between $\mathcal{I}$ and $\mathcal{I}+d \mathcal{I}$.

In \cite{Pennisi_Ruggeri},  using the Maximum Entropy Principle (MEP),  the  authors found the   equilibrium distribution function  that generalizes the J\"uttner one:
\begin{equation}\label{5.2n}
{f_E= \frac{n \gamma}{A(\gamma) K_2(\gamma)} \frac{1}{4 \pi m^3
c^3} e^{- \frac{\gamma}{mc^2} \left[ \left( 1 + \frac{\mathcal{I}}{m c^2} \right) u_\beta
p^\beta \right]}},
\end{equation}
with $A(\gamma)$ given by
\begin{equation*}
A(\gamma)=  \frac{\gamma}{K_2(\gamma)} \int_0^{+\infty} \frac{K_2( \gamma*)}{\gamma*} \,
\phi(\mathcal{I})  \, d \, \mathcal{I},
\end{equation*}
where
\begin{equation*}
\gamma^*= \gamma + \frac{\mathcal{I}}{k_B T}.
\end{equation*}
The pressure and the energy for polyatomic gases, compatible with the distribution function \eqref{5.2n} are \cite{Pennisi_Ruggeri}:
\begin{align}\label{3n}
\begin{split}
&      {p =  \frac{n m c^2}{\gamma} = \frac{k_B}{m} \rho T}  \, , \\
&   {e= \frac{n m c^2}{A(\gamma) K_2(\gamma)} \int_0^{+\infty} \left[ K_3  (\gamma^*)   - \frac{1}{\gamma^*}  K_2
 (\gamma^*) \right]  \phi(\mathcal{I}) \, d \,\mathcal{I}}.
\end{split}
\end{align}
We remark that the pressure has the same expression for  a monatomic and for a polyatomic gas, while   \eqref{3n}$_2$ is the generalization of the Synge energy to the case of polyatomic gases.
The macroscopic internal energy
 in the classical limit, when $\gamma\rightarrow \infty$,
 converges to the one of a classical polyatomic gas  \eqref{state-rarepoly},
 provided that the measure
 \begin{equation*}
 \phi(\mathcal{I}) = \mathcal{I}^a,
 \end{equation*}
 where the constant
 \begin{equation}\label{aaa}
 a= \frac{D-5}{2}.
 \end{equation}

In the ultra-relativistic limit it was proved in \cite{PR2} that the generalized Synge equation \eqref{3n} for a gas with internal structure coincides with the one postulated by Smoller and Temple and other authors but with a precise value of the constant that is related to the degree of freedom
\begin{equation*}
p = \sigma^2 e , \quad \text{with} \quad \sigma^2 =\left\{\begin{array}{cll} \frac{1}{3}, &\quad  & \forall \,\, -1<a \leq 2,\\ \\
\frac{1}{a+1}, & \quad & \forall \,\, a>2, \end{array}\right.
\end{equation*}
with $a$ given by \eqref{aaa}.

For $a \rightarrow -1$, the polyatomic equations converge to the monatomic ones \cite{carrisi}.
The polyatomic gas theory is very complex, but when $a =0$ we will prove that the integral in \eqref{3n} can be written in an analytical way and this case corresponds to the diatomic gas.

\bigskip
As also noted in \cite{Speck-Strain-CMP-2011}, due to complexity of constitutive equations  (\ref{press}), (\ref{caloric}) (monatomic gas) or (\ref{3n}) with $a=0$ (diatomic gas) , basic issues such as the maps' invertibility between fluid dynamic variables which are expressed as functions of any two of them, and hyperbolicity of the relativistic Euler system are difficult to verify. We will show that despite the relativistic Euler system (\ref{main1}) with state relations (\ref{press}), (\ref{caloric}) (monatomic gas) or (\ref{3n}) with $a=0$ (diatomic gas)  is very complicated, similar results as the Riemann problem of classical Euler system can be obtained.

\smallskip

In order to formulate the main result of the Riemann problem, as in \cite{Smoller-1994}, we define%for $\mathbf{u}\equiv (p,v,S) \in \mathbb{R}^3$,
\begin{eqnarray*}
    &\mathcal{S}_i(\mathbf{u}_L)=\left\{(p,v,S):(p,v,S)
~\text{on~i-shock~waves~from}~\mathbf{u}_L\right\},
\\
&\mathcal{R}_i(\mathbf{u}_L)=\left\{(p,v,S):(p,v,S)
~\text{on~i-rarefaction~waves~from}~\mathbf{u}_L\right\},\quad i=1,3,
\end{eqnarray*}
where $\mathbf{u}_L$ is the left point on the $i-th$ wave $\mathcal{S}_i(\mathbf{u}_L)$ or $\mathcal{R}_i(\mathbf{u}_L)$. And we further define
\begin{eqnarray*}
&\mathcal{S}^p_i(\mathbf{u}_L)=\left\{(p,v):(p,v)\in ~\text{Projection~of~}
\mathcal{S}_i(\mathbf{u}_L) \right\}
,\\
&\mathcal{R}^p_i(\mathbf{u}_L)=\left\{(p,v):(p,v)\in ~\text{Projection~of~}
\mathcal{R}_i(\mathbf{u}_L) \right\} ,\\
&\mathcal{T}_i^p(\mathbf{u}_L)=\mathcal{S}^p_i(\mathbf{u}_L)\cup \mathcal{R}_i^p(\mathbf{u}_L),\quad i=1,3.
\end{eqnarray*}
The principal aim of this paper is to prove the following theorem:
\begin{theorem}\label{main theorem}
For the relativistic Euler system (\ref{main1}), let its constitutive equations be given as in (\ref{press}), (\ref{caloric}) (monatomic gas) or (\ref{3n}) with $a=0$ (diatomic gas)  and its Riemann initial data be given in (\ref{ini-data}). Then a vacuum occurs in the solution of the Riemann problem if
$$\bar{r}_L<\bar{s}_R,$$
where $\bar{r}_L$ and $\bar{s}_R$ are the 1-Riemann invariant and 3-Riemann invariant corresponding to the left state $\mathbf{u}_L$ and right state $\mathbf{u}_R$, respectively.
In the opposite case, namely $\bar{r}_L\geq \bar{s}_R$, the Riemann problem admit a unique solution. As in Figure \ref{curves}, the  $p-v$ plane is divided into four parts by the curves $\mathcal{T}_1^p(\mathbf{u}_L)$ and $\mathcal{T}_3^p(\mathbf{u}_L)$.
\end{theorem}
\begin{figure}[ht]
	\centering
	\includegraphics[width=0.8\linewidth]{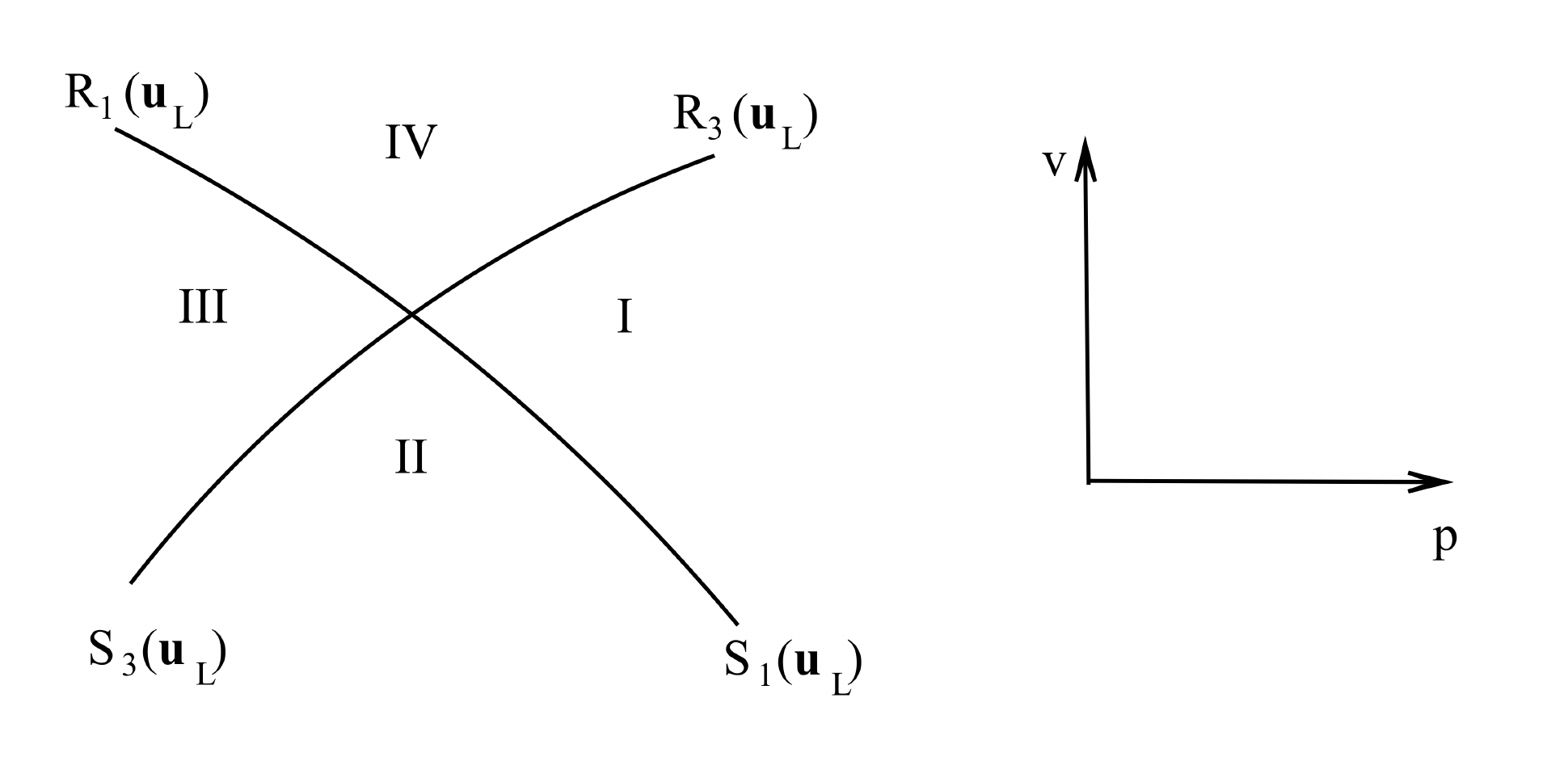}
	\caption{Projected wave curves on the $p-v$ plane.}
	\label{curves}
\end{figure}
Before presenting the problematic of the Riemann problem in these cases, it is convenient to recall first the concept of \emph{entropy principle}, \emph{main field}, \emph{symmetrization} and \emph{entropy growth across the shock} for a general hyperbolic system of balance laws which are essential to the following analysis.

\section{Entropy principle, symmetric form and growth of entropy}\label{Ruggeri-Strumia}
The relativistic system \eqref{massmomenergy} belongs to a general
system of $N$ balance laws for the field $\mathbf{u}(x^\beta) \in \bbr^N$:
\begin{equation}\label{B-1}
\partial_\alpha {\bf F}^\alpha({\bf u})={\bf f}({\bf u}),
\end{equation}
where $\mathbf{F}^\alpha$, $(\alpha=0,1,2,3)$ and $\mathbf{f}$ are column vectors in $\bbr^N$ representing densities-fluxes and production terms, respectively. Now, any theory of continuum needs
to be compatible with the entropy principle which
  requires that system \eqref{B-1} has a natural entropy-entropy flux pair $h^\alpha$  satisfying a supplementary balance law:
\begin{equation}\label{B-2}
\partial_\alpha h^\alpha({\bf u})=\Sigma({\bf u}),
\end{equation}
where $\Sigma $ is the  entropy production term, which is nonnegative according to the second law of thermodynamics. We also assume that $h=h^\alpha \xi_\alpha $ is a convex function of the field $\mathbf{u}=\mathbf{F^\alpha}\xi_\alpha$, where $\xi_\alpha $ is a constant  time-like congruence.
\subsection{Main field and symmetric form}
In \cite{RS81}, Ruggeri and Strumia observed that \eqref{B-1} and \eqref{B-2} form a overdetermined quasilinear hyperbolic system. Thus, in order for any smooth solution of \eqref{B-1} to satisfy the entropy law \eqref{B-2}, the equation \eqref{B-2} must be obtained as a linear combination of the equations of system \eqref{B-1}: there exists a vector $\mathbf{u}'({\bf u}) \in \bbr^N$ such that
\begin{equation}\label{B-3}
\partial_\alpha h^\alpha({\bf u})-\Sigma({\bf u}) \equiv {\bf u}'({\bf u})\cdot(\partial_\alpha {\bf F}^\alpha({\bf u})-{\bf f}({\bf u})).
\end{equation}
Since \eqref{B-3} is an identity, by comparing the differential terms and production  terms, one has:
\begin{equation}\label{B-3-1}
\textup{d}h^\alpha={\bf u}'\cdot \textup{d}{\bf F}^\alpha, \qquad \Sigma= {\bf u}'\cdot {\bf f}\ge0.
\end{equation}
Next, introduce potentials  $h'^\alpha$ defined as follows:
\begin{equation*}
h'^\alpha= {\bf u}'\cdot{\bf F}^\alpha -h^\alpha.
\end{equation*}
Then, it follows from \eqref{B-3-1} that
\[d h'^\alpha=d {\bf u}'\cdot {\bf F}^\alpha.\]
Now, if one chooses ${\bf u}'$ as a new field, one has
\begin{equation}\label{2.6}
{\bf F}^\alpha =\frac{\partial h'^\alpha}{\partial {\bf u}'}.
\end{equation}
Then
\begin{equation} \label{B-3-2}
\partial_\alpha \mathbf{F}^\alpha = \partial_\alpha\left(\frac{\partial h'^\alpha}{\partial {\bf u}'}\right)=\left(\frac{\partial^2 h'^\alpha}{\partial {\bf u}'\partial {\bf u}'}\right)\partial_{\alpha} {\bf u}'.
\end{equation}
Combine \eqref{B-1} and \eqref{B-3-2} to rewrite the original system \eqref{B-1} in the form
\begin{equation}\label{symm}
\left(\frac{\partial^2 h'^\alpha}{\partial {\bf u}'\partial {\bf u}'}\right)\partial_{\alpha} {\bf u}'={\bf f}({\bf u}').
\end{equation}
Since $h'=h'^\alpha \xi_\alpha$ is the Legendre transform of $h= h^\alpha \xi_\alpha $, it follows from \eqref{B-3-1}$_1$ that
\[
dh = d(h^\alpha \xi_\alpha) =  \mathbf{u}' \cdot d \mathbf{u}, \quad \leftrightarrow  \quad \mathbf{u}'=\frac{\partial h} {\partial\mathbf{u}},
\]
i.e. $\mathbf{u}'$ is the dual field of  (multiplying \eqref{2.6} by $\xi_\alpha$)
\[
\mathbf{u}=\frac{\partial h'}{\partial {\mathbf{u}'}}.
\]
We observe that the map $\mathbf{u}'(\mathbf{u}) $ is globall invertible (see \cite{RS81}).
 Then one concludes that the original  system  \eqref{B-1} is expressed as  the form   \eqref{symm} if we choose the field $\mathbf{u}'$. This  is a very  special symmetric  system  according with the Friedrichs definition. In fact all matrices are symmetric and
 \begin{equation*}
 \left(\frac{\partial^2 h'^\alpha}{\partial {\bf u}'\partial {\bf u}'}\right)\xi_\alpha =  \frac{\partial^2 h'}{\partial {\bf u}'\partial {\bf u}'}
 \end{equation*}
 is positive definite.
This result given  in \cite{RS81}  generalizes  Boillat's  symmetrization result  \cite{Boillat} in covariant formalism. This symmetrization holds only for the new field  ${\bf u}'$. This is why this field  ${\bf u}'$ was called  the \emph{main field} by Ruggeri and Strumia \cite{RS81}. The system \eqref{symm} is also frequently called as \emph{Godunov system}, since Godunov was the first one who symmetrizes the Euler system for fluids and physical systems arising from  a variational principle \cite{Godunov}. The interested reader can read a brief history on the symmetrization procedure for a system compatible with an entropy principle in Chapter 2 of \cite{book}.

\subsection{Entropy growth across a shock wave}
%\label{ruggeri-strumia}
Let $\Omega$ be a connected open set of $V^4$ and $ \mathcal{S}$ an hyper-surface cutting $\Omega$ into
two open subsets $\Omega_1$, $\Omega_2$. Let $\phi(x_\alpha) =0 $ be an equation of  $\Omega$ referred to any coordinate frame: we shall identify $\Omega$ with a shock
hyper-surface for the field $\mathbf{u}$.
Then it is known that the Rankine-Hugoniot conditions must hold :
\begin{equation*}
\phi_\alpha \left[[ \mathbf{F}^\alpha(\mathbf{u})\right]] =0,
\end{equation*}
where the square bracket indicates the jump in $\Omega$:
\begin{equation*}
\left[[ \mathbf{F}^\alpha(\mathbf{u})\right]]  = \mathbf{F}^\alpha(\mathbf{u}_L) - \mathbf{F}^\alpha(\mathbf{u}_R).
\end{equation*}
Formally the Rankine-Hugoniot equations are obtained from the field
equations \eqref{B-1} through the correspondence rule
\begin{equation*}
\partial_{\alpha} \quad  \rightarrow \quad \phi_\alpha \left[[ \, \, \, \,\right]], \qquad
\mathbf{f}  \quad  \rightarrow \quad 0.
\end{equation*}
But the previous rule does not hold when it is applied on the supplementary equation \eqref{B-2} since
\begin{equation}\label{etina}
\eta= \phi_\alpha \left[[ h^\alpha\right]],
\end{equation}
is generally non vanishing.
We can decompose $\phi_\alpha = - s\, \xi_\alpha + \zeta_\alpha$ with $\xi_\alpha$ and $\zeta_\alpha$ respectively constant time-like and space-like congruences and $s$ is a shock velocity. Therefore \eqref{etina} becomes
\begin{equation*}
\eta= -s  \left[[ h \right]] +\zeta_\alpha \left[[ h^\alpha\right]].
\end{equation*}
Ruggeri and Strumia proved that if $h$ is convex, then $\eta$ is an increasing function of $s$  and the positive branch (admissible shocks) requires that the shock velocity $s$ is greater than the  corresponding characteristic velocity evaluated in the unperturbed equilibrium state:
\begin{equation*}
\eta >0 \qquad \text{iff} \quad s>\lambda, \quad \text{with} \quad   \lambda=\lambda(\mathbf{u}_R).
\end{equation*}
\subsection{Consequences for  the relativistic Euler system}
In the case of relativistic  Euler fluid \eqref{massmomenergy}, \eqref{EuVT},  the entropy  law is  \eqref{B-2} with
\begin{equation}\label{halfa}
h^\alpha = - \rho S u^\alpha, \qquad \Sigma =0,
\end{equation}
where $S$ is the entropy density satisfying the Gibbs equation:
\begin{equation}\label{Gibbs}
T dS = d \varepsilon -\frac{p}{\rho^2} d\rho.
\end{equation}
The system is symmetric hyperbolic in the
main field \cite{RS81} :
\begin{equation}\label{main_rel}
\mathbf{u}' \equiv \frac{1}{T}   \left(\left(g+c^2\right),-{u_\beta}\right),
\end{equation}
where $g$ is the chemical potential
\begin{equation*}
g=\varepsilon +\frac{p}{\rho} - T S.
\end{equation*}

In the same paper \cite{RS81} (see also \cite{Cime}), it was proved that the convexity of entropy is equivalent that the maximum characteristic velocity in the rest frame satisfies the sub-luminal condition and the specific heat at constant  pressure $c_p$ is positive:
 \begin{equation}\label{condQ}
p_e = \left. \frac{\partial p}{\partial e}\right|_S < 1, \qquad c_p= \frac{k_B}{m} + c_V >0,
\end{equation}
where $c_V = d\varepsilon /d T$ is the specific heat at constant volume.
The two conditions in \eqref{condQ} are equivalent to the hyperbolicity and sub-luminal conditions:
 \begin{equation}\label{convcondi}
0 <p_e  < 1.
\end{equation}
We will prove in the following that for relativistic Euler system with Synge energy, the inequalities  \eqref{convcondi} can be verified. We can conclude that the system of relativistic Euler is symmetric hyperbolic in the main field given by \eqref{main_rel}, the entropy is convex and it grows across a shock wave.
\section{Relativistic rarefied monatomic gas}

% % % % % % %
In this section, we analyze basic properties of the relativistic Euler system (\ref{main1}) with constitutive equations (\ref{press}), (\ref{caloric}) and (\ref{num-den}). Due to the complexity of the relativistic Euler system and the modified Bessel functions, the analysis of basic properties of (\ref{main1}) such as the strict hyperbolicity and genuine nonlinearity  is far from trivial.

First of all, from the property of the Bessel functions Appendix 1 \eqref{transform}, we can rewrite the constitutive equation \eqref{caloric} as the following form
\begin{equation}\label{ener-den}
e=c^2\rho\frac{K_1(\gamma)}{K_2(\gamma)}+3p = p\left(\gamma\frac{K_1(\gamma)}{K_2(\gamma)}+3\right).
\end{equation}
And it is also convenient to write the expression of $n$ as a function of $\gamma$ and entropy density $S$ (see e.g. \cite{Groot-Leeuwen-Weert-1980, Speck-Strain-CMP-2011}):
\begin{equation}\label{num-den}
n=4\pi e^4m^3c^3h^{-3}e^{\frac{-S}{k_B}}\frac{K_2(\gamma)}{\gamma}e^{\gamma\frac{K_1(\gamma)}{K_2(\gamma)}}.
\end{equation}

\subsection{Characteristic velocities}
The relativistic Euler system \eqref{massmomenergy} or  \eqref{main1} is a particular case of a general system of conservation laws:
\begin{equation}\label{main2}
 \partial_t \mathbf{u} +\partial_x  \mathbf{F}(\mathbf{u})=0,
\end{equation}
with
\begin{align}\label{urrel}
\begin{aligned}
\mathbf{u}&\equiv \frac{1}{c}\left(V^0,T^{01}, c \, T^{00}\right) =
\left(
\frac{\rho c}{\sqrt{c^2-v^2}},
\frac{(e+p)v}{c^2-v^2},
\frac{e c^2+pv^2}{c^2-v^2}\right),\\
\mathbf{F}(\mathbf{u})&\equiv\left(V^1,T^{11},c \, T^{10}\right)
=\left(\frac{\rho c v}{\sqrt{c^2-v^2}},
\frac{(e+p)v^2}{c^2-v^2}+p,
\frac{(e+p) c^2 v}{c^2-v^2}\right).
\end{aligned}
\end{align}
	Eigenvalues of (\ref{main2})  are
	\begin{align}\label{lamdas}
	\begin{aligned}
	&\lambda_1=\frac{(e_{p}(p,S)-1)c^2v-\sqrt{e_{p}(p,S)}c(c^2-v^2)}{e_{p}(p,S)c^2-v^2},\\
	&\lambda_2=v,\\
	&\lambda_3=\frac{(e_{p}(p,S)-1)c^2v+\sqrt{e_{p}(p,S)}c(c^2-v^2)}{e_{p}(p,S)c^2-v^2},
	\end{aligned}
	\end{align}
		with $e_p =1/p_e$.
Eigenvectors $r_i, (i=1, 2, 3)$  corresponding to $\lambda_i$ are
	\begin{align} \label{eigen3}
	\begin{aligned}
	&r_1=\widetilde{r}\left(1,-\frac{\sqrt{e_{p}}(c^2-v^2)}{(e+p)c},0\right),\\
	&r_2=(0,0,1),\\
	&r_3=\widetilde{r}\left(1,\frac{\sqrt{e_{p}}(c^2-v^2)}{(e+p)c},0\right),
	\end{aligned}
	\end{align}
	where $\widetilde{r}$ is an arbitrary scalar function to be determined later.
	Denote $\bar{\lambda} =\lambda/c$ as
	the characteristic velocities \eqref{lamdas} in the unity of light velocity and let $\hat{\lambda}=\bar{\lambda}|_{v=0}$, i.e. the characteristic velocity in the unity of light velocity evaluated in the rest frame. Then, from \eqref{lamdas} and reference \cite{PR2}, we have
	\begin{align}\label{case44}
	\begin{split}
	& \hat{\lambda}_2 =0, \quad  \hat{\lambda}_{1,3}  =\mp \frac{1}{\sqrt{e_p}} =\mp\sqrt{p_e} = \mp \sqrt{ \frac{r+1-r ^\prime \gamma}{(r+1)(r- r^\prime \gamma)}} =\mp \sqrt{\frac{c_p}{c_V}\frac{p}{p+e}},
	  \\
	& \text{with }\quad     r= \frac{e}{p} \quad \mbox{and}  \quad r ^\prime = \frac{dr(\gamma)}{d\gamma}.
	\end{split}
	\end{align}
 % We use  (\ref{temp}) and (\ref{caloric}) to have
%\begin{align}\label{cV1}
%\begin{aligned}
%\frac{d\gamma}{dT}=&-\frac{\gamma}{T}=-\frac{k_B\gamma^2}{mc^2},\\
%\varepsilon=&\frac{ c^2}{  K_2(\gamma)}  \left[ K_3 \left(\gamma
%\right) - \frac{1}{\gamma}  K_2- 1\right]=c^2\left[\frac{K_1 \left(\gamma
%\right)}{K_2 \left(\gamma
%\right)}+\frac{3}{\gamma}-1\right].
%\end{aligned}
%\end{align}
%Then we use (\ref{cV1}) and Appendix (\ref{imp-ine1}) to have
%\begin{align*}
%\begin{aligned}
%c_V=&\frac{\partial \varepsilon}{\partial T}=c^2\frac{\partial}{\partial \gamma}\left[\frac{K_1 \left(\gamma
%\right)}{K_2 \left(\gamma
%\right)}+\frac{3}{\gamma}-1\right]\frac{\partial\gamma}{\partial T},\\
%=&-\frac{k_B}{  m}  \left[\gamma^2\left(\frac{K_1(\gamma)}{K_2(\gamma)}\right)^2+3\gamma\frac{K_1(\gamma)}{K_2(\gamma)}- \gamma^2-3\right]>0.
%\end{aligned}
%\end{align*}
%That is, (\ref{cV}) holds.

%%%%%%%%%%%%%%%%%%%%%%%
\subsection{Strict hyperbolicity}
This part is devoted to the proof of the strict hyperbolicity for the system (\ref{main2}). In fact, we have the following proposition:
\begin{proposition} The system \eqref{main2} with constitutive equations \eqref{press}, \eqref{ener-den} and (\ref{num-den})   is strictly hyperbolic.
	%\begin{equation}\label{hyper00}
%	\lambda_-=\lambda_1<v=\lambda_2<\lambda_+=\lambda_3.
%	\end{equation}
\end{proposition}
\begin{proof}
First we prove that
 \begin {equation}\label {cV}
c_V=\frac{ d\varepsilon}{d T} >0.
 \end{equation}
In fact, from   \eqref{energia-e} and \eqref{temp}, we  have
\begin{align}\label{cV1}
\begin{aligned}
\frac{d\gamma}{dT}=-\frac{\gamma}{T},\qquad
\varepsilon= \rho c^2\left[\frac{K_1 \left(\gamma
\right)}{K_2 \left(\gamma
\right)}+\frac{3}{\gamma}-1\right].
\end{aligned}
\end{align}
Then we use (\ref{cV1}) and Appendix 3 (\ref{imp-ine1}) to have
\begin{align*}
\begin{aligned}
c_V=&\frac{\partial \varepsilon}{\partial T}=\rho c^2\frac{\partial}{\partial \gamma}\left[\frac{K_1 \left(\gamma
\right)}{K_2 \left(\gamma
\right)}+\frac{3}{\gamma}-1\right]\frac{\partial\gamma}{\partial T},\\
=&-\frac{k_B\rho }{  m}  \left[\gamma^2\left(\frac{K_1(\gamma)}{K_2(\gamma)}\right)^2+3\gamma\frac{K_1(\gamma)}{K_2(\gamma)}- \gamma^2-3\right]>0.
\end{aligned}
\end{align*}
 Then the strict hyperbolicity of the system \eqref{main2} with constitutive equations \eqref{press}, \eqref{ener-den} and (\ref{num-den}) follow from \eqref{condQ} and \eqref{case44}.
\end{proof}

%%%%%%%%%%%%%%%%%%%%%
%
%
%%%%%%%%%%%%%%%%%%%%%%%
\subsection{Sub-luminal characteristic velocities}
We want to prove in this subsection that characteristic velocities of the system (\ref{main2}) are sub-luminal.
\begin{theorem}\label{sub-luminal}  For the system (\ref{main2}) with relations (\ref{press}), (\ref{ener-den}) and (\ref{num-den}), we have
	\begin{equation}\label{max-lmd}
\bar{\lambda}_{max}=	\bar{\lambda}_3(\gamma,v)<1,\qquad \hat{\lambda}_{max} =   \bar{	\lambda}_3(\gamma,0)<\frac{1}{\sqrt{3}}.
	\end{equation}

\end{theorem}
\begin{proof}
			From (\ref{press}), \eqref{ener-den} and (\ref{num-den}), we obtain
	\begin{equation}\label{p0}
	\begin{split}
	&p=4\pi e^4m^4c^5h^{-3}e^{\frac{-S}{k_B}}\frac{K_2(\gamma)}{\gamma^2}e^{\gamma\frac{K_1(\gamma)}{K_2(\gamma)}},\\
	&e=p\left(\gamma\frac{K_1(\gamma)}{K_2(\gamma)}+3\right).
	\end{split}
	\end{equation}
	We use Appendix 1 (\ref{transform}) and (\ref{deriva}) to have
	\begin{equation}\label{p1}
	\begin{split}
	&\frac{d}{d\gamma}\left(\frac{K_2(\gamma)}{\gamma^2}\right)=-\frac{4K_2(\gamma)+\gamma K_1(\gamma)}{\gamma^3},\\
	 &\frac{d}{d\gamma}\left(\gamma\frac{K_1(\gamma)}{K_2(\gamma)}\right)=\gamma\left(\frac{K_1(\gamma)}{K_2(\gamma)}\right)^2+4\frac{K_1(\gamma)}{K_2(\gamma)}-\gamma.
	\end{split}
	\end{equation}
	Applying (\ref{p1}) in (\ref{p0}) yields
	
	$$\frac{\partial_\gamma p}{p}=\gamma\left(\frac{K_1(\gamma)}{K_2(\gamma)}\right)^2+3\frac{K_1(\gamma)}{K_2(\gamma)}-\gamma-\frac{4}{\gamma}, \qquad\mbox{and}$$
	\begin{equation}\label{pz}
	\begin{split}
	e_{p}&=\frac{e}{p}+\frac{p}{\partial_\gamma p}\frac{d}{d\gamma}\left(\gamma\frac{K_1(\gamma)}{K_2(\gamma)}\right)\\
	 &=3+\gamma\frac{K_1(\gamma)}{K_2(\gamma)}+\frac{\gamma\left(\frac{K_1(\gamma)}{K_2(\gamma)}\right)^2+4\frac{K_1(\gamma)}{K_2(\gamma)}-\gamma}{\gamma\left(\frac{K_1(\gamma)}{K_2(\gamma)}\right)^2+3\frac{K_1(\gamma)}{K_2(\gamma)}-\gamma-\frac{4}{\gamma}}\\
	 &=3+\frac{\gamma^2\left(\frac{K_1(\gamma)}{K_2(\gamma)}\right)^3+4\gamma\left(\frac{K_1(\gamma)}{K_2(\gamma)}\right)^2-\gamma^2\frac{K_1(\gamma)}{K_2(\gamma)}-\gamma}{\gamma\left(\frac{K_1(\gamma)}{K_2(\gamma)}\right)^2+3\frac{K_1(\gamma)}{K_2(\gamma)}-\gamma-\frac{4}{\gamma}}>3,
	\end{split}
	\end{equation}
	by  Appendix 3 (\ref{imp-ine1}) and (\ref{imp-ine2}). Then, from \eqref{lamdas} and \eqref{pz}, we have
$$\bar{\lambda}_{max} -1=\frac{(\sqrt{e_{p}}-1)(v-\sqrt{e_p}c)(c-v)}{e_{p}c^2-v^2}<0.$$
 Then we obtain the first inequality in (\ref{max-lmd}).
Moreover, we use (\ref{pz}) to have
 $$\hat{\lambda}_3=\frac{1}{\sqrt{e_p}}=\sqrt{p_e}<\frac{1}{\sqrt{3}}.$$
%On the other hand,  we have
%$$\frac{\partial \bar{\lambda}_3(p, S)}{\partial v}=\frac{(e_{p}(p, S)-1)(v-\sqrt{e_{p}}c)^2c}{(e_{p}c^2(p, S)-v^2)^2}>0.$$
\end{proof}	
\begin{figure}[ht]
	\centering
	\includegraphics[width=0.6\linewidth,angle=90]{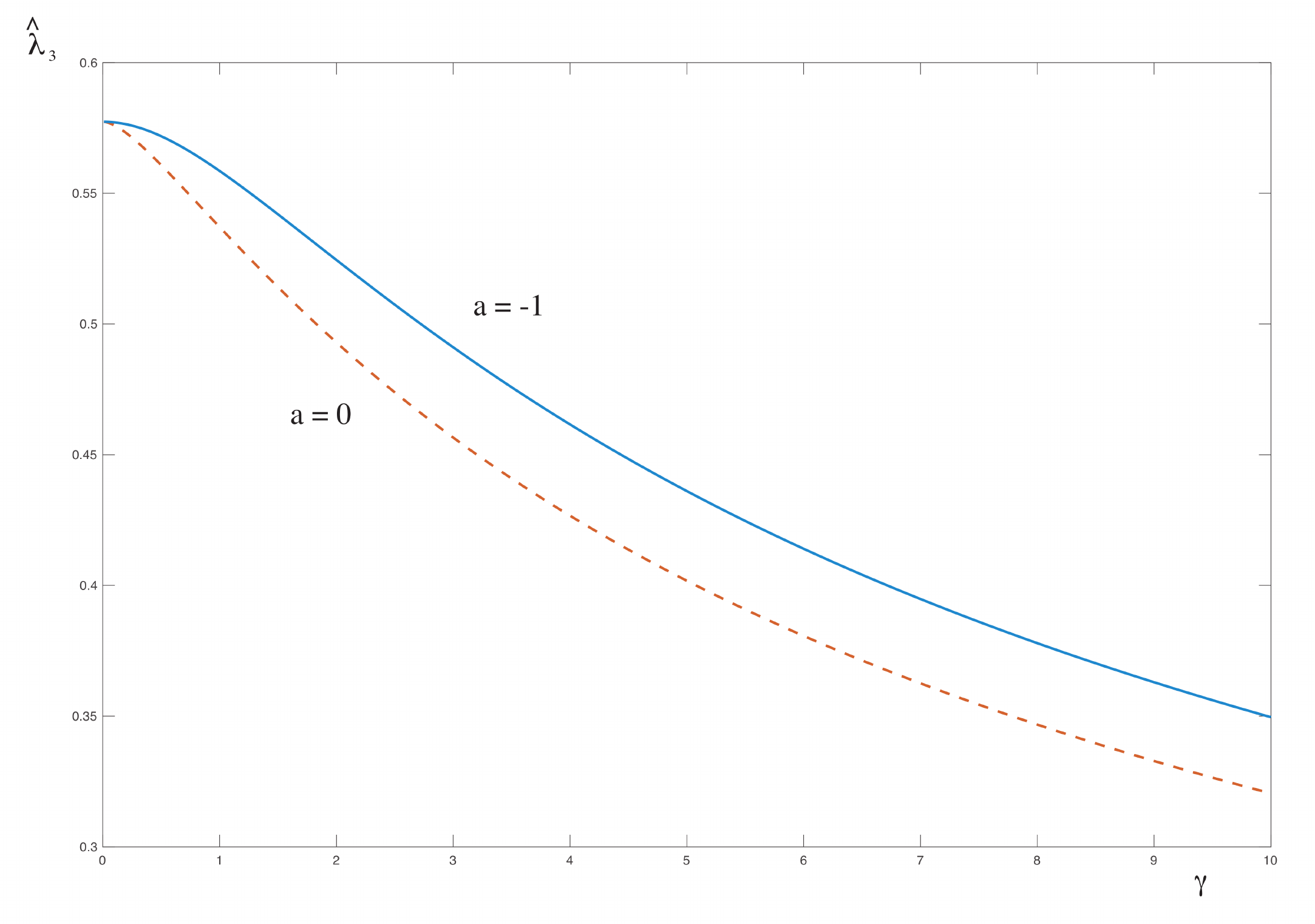}
		\caption{The behavior of $\hat{	\lambda}_3$ versus $\gamma$ for $a=-1$ (Monatomic gas) and for $a=0$ (Diatomic gas).}
	\label{fig:lambdawuhan}
\end{figure}
From \eqref{case44}, we can plot  $\hat{\lambda}_3$ versus $\gamma$   and we see also  from   Figure \ref{fig:lambdawuhan} (with $a=-1$) that in the ultra-relativistic limit ($\gamma \rightarrow 0$), this velocity tends to $1/\sqrt{3}$ and decays monotonically with respect to $\gamma$ .

Therefore we proved that the inequality \eqref{convcondi} can be verified and from the general results of \cite{RS81} addressed in Section \ref{Ruggeri-Strumia}, the following theorem holds:
\begin{theorem}
The system \eqref{main2} satisfies the entropy law \eqref{B-2} with \eqref{halfa} in space-time form:
\begin{equation}\label{entra-cons}
\partial_t \left(\rho S \Gamma \right) + \partial_x \left(\rho S \Gamma v \right)=0.
\end{equation}
The system is symmetric hyperbolic with respect to the main field \eqref{main_rel}:
\begin{align}
\begin{aligned}
\mathbf{u}'&\equiv \frac{1}{T}\left(\frac{e+p}{\rho}-T S, \Gamma v, \frac{\Gamma}{c}\right).\notag
\end{aligned}
\end{align}
The entropy density
\begin{equation*}
h= -  \rho S \Gamma
\end{equation*}
is a convex function with respect to the field  $\mathbf{u}$ given by \eqref{urrel}, and the entropy grows across the shock.
\end{theorem}

\subsection{Genuine nonlinearity}\label{sec-genui}
In this subsection, we show that $\lambda_2=v$ is linearly degenerate, while $\lambda_1$ and $\lambda_3$ are genuinely nonlinear.

Denote $\nabla\lambda_2$ as the gradient of $\lambda_2$ with respect to $(p, v, S)$. Then $\lambda_2=v$ is linearly degenerate since
\begin{equation}\label{eigvec2}
\nabla\lambda_2\cdot r_2=(0,1,0)\cdot(0,0,1)^T=0.
\end{equation}

We now turn to show that $\lambda_1$ is genuinely nonlinear. Since the genuine nonlinearity of $\lambda_3$  can be done similarly, we omit the details.
The eigenvalue $\lambda_1$ satisfies
\begin{align*}
\begin{aligned}
&\nabla\lambda_1=\left(\frac{\partial\lambda_1}{\partial p},\frac{\partial\lambda_1}{\partial v},\frac{\partial\lambda_1}{\partial S}\right) = \nonumber\\
& \left(\frac{e_{pp}c(c^2-v^2)}{2\sqrt{e_{p}}(e_{p}c^2-v^2)^2}(v+\sqrt{e_{p}}c)^2,\frac{(e_{p}-1)c^2(v+\sqrt{e_{p}}c)^2}{(e_{p}c^2-v^2)^2}, \right. \\
& \qquad \left. \frac{e_{pS}c(c^2-v^2)}{2\sqrt{e_{p}}(e_{p}c^2-v^2)^2}(v+\sqrt{e_{p}}c)^2\right).
\end{aligned}
\end{align*}
Then taking into account of \eqref{eigen3}$_1$,  we have
\begin{equation}\label{genui}
\begin{split}
&\nabla\lambda_1\cdot r_1 =\\
&\frac{\widetilde{r}e_{pp}c(c^2-v^2)}{2\sqrt{e_{p}}(e_{p}c^2-v^2)^2}(\sqrt{e_{p}}c+v)^2-\frac{\widetilde{r}\sqrt{e_{p}}(e_{p}-1)c(c^2-v^2)}
{(e+p)(e_{p}c^2-v^2)^2}(\sqrt{e_{p}}c+v)^2 =\\
& \qquad \frac{\widetilde{r}(\sqrt{e_{p}}c+v)^2c(c^2-v^2)}{2(e+p)\sqrt{e_{p}}(e_{p}c^2-v^2)^2}\left[(e+p)e_{pp}-2e_{p}(e_{p}-1)\right].
\end{split}
\end{equation}

\begin{proposition}\label{genuine} For any $\gamma\in(0,\infty)$, it holds that
	\begin{equation}\label{negativity}
	(e+p)e_{pp}-2e_{p}(e_{p}-1)<0.
	\end{equation}
	Then the eigenvalue  $\lambda_1$ is genuinely nonlinear and we can choose the function $\widetilde{r}$ properly such that:
	\begin{equation}\label{genuine0}
	\begin{split}
	\nabla\lambda_1\cdot r_1=1.
	\end{split}
	\end{equation}
	
\end{proposition}
Since the proof of Proposition \ref{genuine} is quite long, we put it in Appendix 4.

%%%%%%%%%%%%%%%%%%%%%
%
%
%%%%%%%%%%%%%%%%%%%%%%%
\subsection{Riemann invariants}
In this part, we solve Riemann invariants for each eigenvalue $\lambda_i, (i=1,2,3)$.

Noting the eigenvector of $\lambda_2$ is $(0,0,1)$ in (\ref{eigvec2}), it is trivial to find the two Riemann invariants of $\lambda_2$ are
$$v \quad\mbox{and}\quad p.$$
We now solve the Riemann invariants for $\lambda_1$. According to the definition, the corresponding Riemann invariant $w$ satisfies
\begin{align}\label{Riemann}
\begin{aligned}
0=&(w_p,w_v,w_{S})\cdot r_1\\
=&\tilde{r}_1(w_p,w_v,w_{S})\cdot\left(1,-\frac{\sqrt{e_{p}}(c^2-v^2)}{(e+p)c},0\right)\\
=&\tilde{r}_1\Big[w_p-\frac{\sqrt{e_{p}}(c^2-v^2)}{(e+p)c}w_v\Big].
\end{aligned}
\end{align}
It is straightforward to see that $S$ is one of the Riemann invariants satisfying (\ref{Riemann}). We further solve (\ref{Riemann}) to find that the other Riemann invariant should be constant along the curve determined by
\begin{equation}\label{1sp}
\frac{\sqrt{e_{p}}dp}{(e+p)c}=-\frac{dv}{c^2-v^2}.
\end{equation}
Solving this differential equation to get the other Riemann invariant $\frac12\ln\Big(\frac{c+v}{c-v}\Big)+\int\frac{\sqrt{e_{p}}dp}{(e+p)c}$. Thus the two Riemann invariants of $\lambda_1$ are
$$S\quad\mbox{and}\quad \bar{r}=\frac12\ln\Big(\frac{c+v}{c-v}\Big)+\int_0^p\frac{\sqrt{e_{p}}dp}{(e+p)c}.$$

Similarly, the two Riemann invariants of $\lambda_3$ are
$$S\quad\mbox{and}\quad \bar{s}=\frac12\ln\Big(\frac{c+v}{c-v}\Big)-\int_0^p\frac{\sqrt{e_{p}}dp}{(e+p)c}.$$

%%%%%%%%%%%%%%%%%%%%%%%%%%%%%%%%%
%
%
%%%%%%%%%%%%%%%%%%%%%%%%%%%%%%%%%

%%
%%%%%%%%%%%%%%%%%%%%%%%%%%%%%%%%%%
%\subsection{Riemann Problem for monatomic gas}
%In this section, we are devoted to the solvability of the Riemann problem (\ref{main2}) and (\ref{ini-data}). For this purpose, the properties of the shock waves and the rarefaction waves and the condition for a vacuum to occur will be invested with the preparation in Section 2 and 3. Then we prove our main result Theorem \ref{main theorem}.

%%%%%%%%%%%%%%%%%%%%%%%%%%%%%%%%%
%
%
%%%%%%%%%%%%%%%%%%%%%%%%%%%%%%%%%
\subsection{Structure of the shock curves}
In this part, we study the structure of the shock curves. It is divided into four parts: the rigorous derivation of the Hugoniot curve; verification of the Lax entropy conditions; monotonicity of the entropy along the shock curves; monotonicity of the velocity along the shock curves.
%%%%%%%%%%%%%%%%%%%%%%%%%%%%%%%%%
%
%
%%%%%%%%%%%%%%%%%%%%%%%%%%%%%%%%%
\subsubsection{Hugoniot curve} We first rigorously derive the Hugoniot curve of the shock curves.
\begin{proposition} Denote $(n_L, e_L, p_L, v_L)$ as the  proper number density, proper energy density, pressure and the velocity of the fluid in the left of the shock curve and  $(n, e, p,v)$ as the corresponding variables at right. Then it holds that
\begin{equation}\label{curve}
\frac{e+p}{n^2}(e+p_L)=\frac{e_L+p_L}{n_L^2}(e_L+p).
\end{equation}
\end{proposition}
\begin{remark}
It should be pointed out that (\ref{curve}) has been proved in
\cite{Taub-1967} and \cite{Chen-ARMA-1997}, and we give a different proof here.
 %the proof is given for the special
%case : shock speed is \gammaero. As far as the authors known, the complete proof of (\ref{curve}) for arbitrary shock speed seems not available. On the other hand,  in our later investigation of the shock curves, we always choose a proper coordinate system such that $v_L=0$. Generally speaking, we can not choose a coordinate system where the shock speed and $v_L$ take the value 0 at the same time. Therefore, we rigorously prove (\ref{curve}) for general shock speed here.
\end{remark}
\begin{proof} Let $s$ be the shock speed. According to the Rankine-Hugoniot conditions, it holds that
\begin{align} \label{RH}
\begin{aligned}
& s\Big(\frac{c\rho}{\sqrt{c^2-v^2}}-\frac{c\rho_L}{\sqrt{c^2-v^2_L}}\Big)=\frac{c\rho v}{\sqrt{c^2-v^2}}-\frac{c\rho_Lv_L}{\sqrt{c^2-v^2_L}}, \\
& s\Big[\frac{(e+p)v}{c^2-v^2}-\frac{(e_L+p_L)v_L}{c^2-v^2_L}\Big]=\frac{(e+p)v^2}{c^2-v^2}+p-\frac{(e_L+p_L)v^2_L}{c^2-v^2_L}-p_L,\\
& s\Big[\frac{(e+p)v^2}{c^2-v^2}+e-\frac{(e_L+p_L)v^2_L}{c^2-v^2}-e_L\Big]= \\ & \frac{(e+p)c^2v}{c^2-v^2}-\frac{(e_L+p_L)c^2v_L}{c^2-v^2_L}.
 \end{aligned}
\end{align}
From (\ref{RH}), it is straightforward to get
\begin{align} %\label{RH0}
\begin{aligned}
& n\Big[\frac{p-p_L}{\sqrt{c^2-v^2}}-\frac{(e_L+p_L)v_L(v_L-v)}{\sqrt{c^2-v^2}(c^2-v^2_L)}\Big]=\\
&n_L\Big[\frac{p-p_L}{\sqrt{c^2-v^2_L}}-\frac{(e+p)(v_L-v)v}{\sqrt{c^2-v^2_L}(c^2-v^2)}\Big], \nonumber\\
& n\Big[\frac{(p+e_L)v}{\sqrt{c^2-v^2}}-\frac{(e_L+p_L)v_L(c^2-vv_L)}{\sqrt{c^2-v^2}(c^2-v^2_L)}\Big]=\\
&n_L\Big[-\frac{(e+p_L)v_L}{\sqrt{c^2-v^2_L}}+\frac{v(e+p)(c^2-v_Lv)}{\sqrt{c^2-v^2_L}(c^2-v^2)}\Big].
 \end{aligned}
\end{align}
Namely,
\begin{align}\label{n1}
\begin{aligned}
&n[(p-p_L)(c^2-v^2_L)+(e_L+p_L)v_L(v-v_L)]\sqrt{c^2-v^2}=\\
&n_L[(p-p_L)(c^2-v^2)+(e+p)v(v-v_L)]\sqrt{c^2-v^2_L},
\end{aligned}
\end{align}
\begin{align}\label{n2}
\begin{aligned}
&n[(p+e_L)v(c^2-v^2_L)-(e_L+p_L)v_L(c^2-vv_L)]\sqrt{c^2-v^2}=\\
&n_L[-(e+p_L) v_L(c^2-v^2)+(e+p)v(c^2-vv_L)]\sqrt{c^2-v^2_L}.
\end{aligned}
\end{align}
Applying $(\ref{n2})-v(\ref{n1})$ and $(\ref{n2})-v_L(\ref{n1})$, we have
\begin{align} \label{RH1}
\begin{aligned}
& n(e_L+p_L)(v-v_L)(c^2-vv_L)\sqrt{c^2-v^2}=\\
&n_L(e+p_L)(v-v_L)(c^2-v^2)\sqrt{c^2-v^2_L}, \\
& n(p+e_L)(v-v_L)(c^2-v^2_L){\sqrt{c^2-v^2}}=\\
&n_L(e+p)(v-v_L)(c^2-vv_L)\sqrt{c^2-v^2_L}.
 \end{aligned}
\end{align}
Note that
$$(v-v_L)(c^2-vv_L)\sqrt{c^2-v^2_L}\sqrt{c^2-v^2}<0.$$
(\ref{RH1}) can be further simplified as
\begin{align} \label{RH2}
\begin{aligned}
& n(e_L+p_L)(c^2-vv_L)=
n_L(e+p_L)\sqrt{c^2-v^2}\sqrt{c^2-v^2_L}, \\
& n(p+e_L)\sqrt{c^2-v^2_L}{\sqrt{c^2-v^2}}=
n_L(e+p)(c^2-vv_L).
 \end{aligned}
\end{align}
We multiply $(\ref{RH2})_1$ by $(\ref{RH2})_2$ and divide the resulting equation by $(c^2-vv_L)\sqrt{c^2-v^2_L}\sqrt{c^2-v^2}$ to have
$$n^2(e_L+p_L)(p+e_L)=n^2_L(e+p_L)(e+p).$$
Then (\ref{curve}) follows.
\end{proof}

%%%%%%%%%%%%%%%%%%%%%%%%%%%%%%%%%
%
%
%%%%%%%%%%%%%%%%%%%%%%%%%%%%%%%%%
\subsubsection{Lax entropy conditions} In this subsection, we show that similar to the non-relativistic Euler system \cite{Smoller-1994,Weinberg-1972}, the Lax entropy conditions are satisfied globally for the shock curves.

\begin{proposition} The Lax entropy conditions hold wholly along the shock curves for the relativistic Euler system (\ref{main1}) with constitutive equations (\ref{press}), (\ref{ener-den}) and (\ref{num-den}). Namely, for a shock curve $\mathbf{u}=\mathbf{u}(\epsilon), \epsilon\leq0$ with shock speed $s=s(\epsilon)$, one has
$$\lambda(\epsilon)<s(\epsilon)<\lambda(0),\hspace{1cm}\epsilon<0.$$
\end{proposition}
\begin{proof}
For brevity, we only prove the inequality $\lambda(\epsilon)<s(\epsilon)$ for the 1-shock curves since the remaining parts can be proven in a similar way. Our proof will be done by contradiction.

Assume $\epsilon_0$ be the first point such that $\lambda_1(\epsilon)=s(\epsilon), \epsilon<0$. Corresponding to our system given in the form (\ref{main2}), the jump condition is $s[[\mathbf{u}]]=[[\mathbf{F}(\mathbf{u})]]$. We differentiate it with respect to $\epsilon$, and multiply the resulted system by the left eigenvector $\ell_1(\epsilon_0)$ at $\epsilon=\epsilon_0$ to have
\begin{equation}\label{jump1}
\begin{split}
&s'[[\mathbf{u}]]+s \mathbf{u}'= dF \mathbf{u}',\\
&s'\ell_1\cdot[[\mathbf{u}]]=(\lambda_1-s)\ell_1\cdot \mathbf{u}'=0.
\end{split}
\end{equation}
$(\ref{jump1})_2$ implies $s'(\epsilon_0)=0$ or $\ell_1\cdot[[\mathbf{u}]]=0$. Now suppose $s'(\epsilon_0)=0$. Then it follows from $(\ref{jump1})_1$ that $\mathbf{u}'(\epsilon_0)=r_1(\epsilon_0)$, and
\begin{equation}\label{jump2}
(s-\lambda_1)'(\epsilon_0)=-\nabla\lambda_1(\epsilon_0)\cdot r_1(\epsilon_0)=-1.
\end{equation}
Since $(s-\lambda_1)'(0)<0$, (\ref{jump2}) implies that there exists some point $\overline{\epsilon}\in (-\epsilon_0, 0)$ such that $\lambda_1(\overline{\epsilon})=s(\overline{\epsilon})$. This contradicts with our choice of $\epsilon_0$. Therefore, our proof can be completed if we can show that
\begin{equation}\label{jump23}
\ell_1\cdot[[\mathbf{u}]]\neq0.
\end{equation}
Now we turn to proof of (\ref{jump23}). For the system (\ref{main2}) with respect to $(n, v, p)$, we can choose the left eigenvector $\ell_1= (\ell_{11}, \ell_{12}, \ell_{13})$ corresponding to $\lambda_1$ in (\ref{eigen3}):
\begin{align*}
\begin{aligned}
\ell_{11}&=\frac{(e+p)(e-pe_{p}(p, S))\sqrt{c^2-v^2}}{\rho},\\
\ell_{12}&=\left[(e+p-pe_{p})v+p\sqrt{e_{p}(p, S)}c\right]c,\\
\ell_{13}&=-\left[(e+p-pe_{p})c+p\sqrt{e_{p}(p, S)}v\right]
\end{aligned}
\end{align*}
 We further choose a coordinate system in which $v_L=0$. Note that
\begin{equation}\label{10shock}
\begin{split}
&\ell_1\cdot[[\mathbf{u}]]=
\frac{(e+p)(e-pe_{p}(p, S))\sqrt{c^2-v^2}}{\rho}\Big(\frac{\rho c}{\sqrt{c^2-v^2}}-\rho_L\Big) +\\
& \Big[(e+p-pe_{p}(p, S))v+p\sqrt{e_{p}(p, S)}c\Big]\frac{(e+p)cv}{c^2-v^2}-\\
& \Big[(e+p-pe_{p}(p, S))c+p\sqrt{e_{p}(p, S)}v\Big]\Big[\frac{(e+p)v^2}{c^2-v^2}+e-e_L\Big] = \\
&(e+p)(e-pe_{p}(p, S))\Big(c-\frac{\rho_L}{\rho}\sqrt{c^2-v^2}\Big)-
\\
&(e+p-pe_{p}(p, S))(e-e_L)c
+p\sqrt{e_{p}(p, S)}[(e+p)-(e-e_L)]v.
\end{split}
\end{equation}
On the other hand, we let $v_L=0$ in $(\ref{RH})_2$ and  $(\ref{RH})_3$ to have
\begin{equation*}
\begin{split}
&\frac{(e+p)sv}{c^2-v^2}=\frac{(e+p)v^2}{c^2-v^2}+p-p_L,\\
& s\Big[\frac{(e+p)v^2}{c^2-v^2}+e-e_L\Big]=\frac{(e+p)c^2v}{c^2-v^2}.
\end{split}
\end{equation*}
Furthermore, we can get
\begin{equation}\label{shock-dis}
s(e-e_L)=\frac{(e+p)v(c^2-v^2)}{c^2-v^2}-(p-p_L)v=(e+p_L)v.
\end{equation}
And we let $v_L=0$ in $(\ref{RH2})$ to have
\begin{equation}\label{shock}
\begin{split}
\frac{n_L\sqrt{c^2-v^2}}{nc}=\frac{e_L+p_L}{e+p_{L}},\qquad\mbox{and}\qquad v^2=\frac{(p-p_L)(e-e_L)}{(p+e_L)(p_L+e)}c^2.
\end{split}
\end{equation}
Note that $s<0$ and $v<v_L=0$. We can use (\ref{shock-dis}) and $(\ref{shock})$ to have
 $$p>p_L,\qquad e>e_L \qquad \mbox{ on the 1-shock curves}.$$
 Then we combine (\ref{10shock}) and (\ref{shock}) to obtain
\begin{equation*}
\begin{split}
&\ell_1\cdot[[\mathbf{u}]]=\\
&\frac{(e-e_L)c}{e+p_L}\Big[(e-pe_{p}(p, S))(e+p)-(e+p-pe_{p}(p, S))(e+p_L)\Big]\\
&+p(p+e_L)\sqrt{e_{p}(p, S)}v<\\
&\frac{(e-e_L)c}{e+p_L}\Big[-(p-p_L)pe_{p}(p, S)-(e+p)p_L\Big]<0.
\end{split}
\end{equation*}
\end{proof}

%%%%%%%%%%%%%%%%%%%%%%%%%%%%%%%%%
%
%
%%%%%%%%%%%%%%%%%%%%%%%%%%%%%%%%%
\subsubsection{Entropy Growth across the shock waves}
Taking into account the results of Section \ref{Ruggeri-Strumia}, and the proof that the characteristic velocities are sub-luminal,  we conclude that:
\begin{corollary}
	For the system of relativistic Euler fluid with Synge energy, the entropy grows across the shock waves.
\end{corollary}
Taking into account the Gibbs equation \eqref{Gibbs},
 the constitutive equation \eqref{press} and the expression of $e$ given in \eqref{energia-e}, we have
 \begin{equation*}
 \frac{m}{k_B} d S = \left(r'(\gamma) - \frac{r(\gamma)}{\gamma}\right) d\gamma - \frac{d\rho}{\rho},
 \end{equation*}
 with
 \begin{equation*}
 r=\frac{e}{p} = \gamma \frac{K_3(\gamma)}{K_2(\gamma)} -1.
 \end{equation*}
Then
 \begin{equation}\label{entropias}
 \frac{m}{k_B} S = \gamma\frac{ K_3(\gamma)}{K_2(\gamma)}  -\ln \gamma + \ln K_2(\gamma) - \ln \rho + \text{const.}
 \end{equation}
Every classical solution of the differential system \eqref{massmomenergy}, thanks to the Gibbs equation \eqref{Gibbs},  also satisfies the supplementary entropy law \eqref{B-2} with \eqref{halfa}, which is expressed as the form
(\ref{entra-cons}) in two dimensional  space-time.
Without loss of generality, we let $v_L=0$. Then the entropy production along the shock \eqref{etina} becomes now
 \begin{equation*}
\eta = (-s+v) \rho \Gamma S + s \rho_L S_L>0.
\end{equation*}
And taking into account the first RH condition of equations \eqref{main1}, we have
 \begin{equation}\label{S13}
\eta = -s\rho_L ( S-S_L)>0\qquad \forall s>\lambda.
\end{equation}
 This is exact what we want to prove. Notice that for the 1-shock curves, $s<0$ and $S>S_L$; while for the 3-shock curves, $s>0$ and $S<S_L$.
\begin{figure}[ht]
\centering
\includegraphics[width=0.9\linewidth]{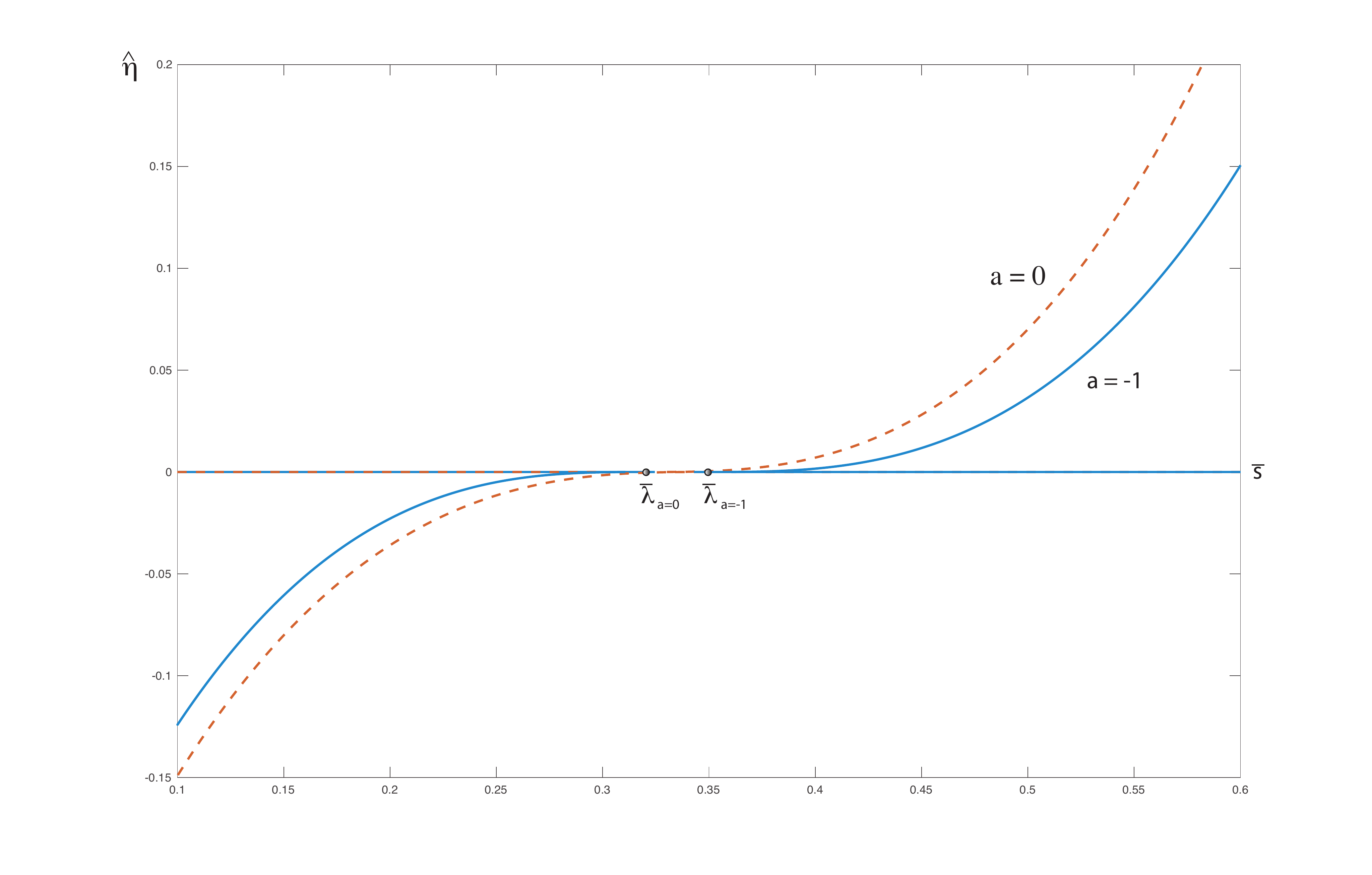}
\caption{Growth of entropy across the shock ($a=-1$: Monatomic gas, $a=0$: Diatomic gas).}
\label{fig:etarel}
\end{figure}
Taking into account \eqref{entropias}, we have definitively
\begin{equation*}
\hat{\eta} = \frac{m}{c k_B\rho_L}\eta =-\bar{s}\left\{
r(\gamma) - r(\gamma_L) + \ln \left(\frac{K_2(\gamma)}{K_2(\gamma_L)} \frac{\gamma_L}{\gamma}\frac{\rho_L}{\rho}\right)\right\}>0.
\end{equation*}
By numerical solution of the RH equations, we can plot $\hat{\eta}$ as a function of $\bar{s} = s/c$. The figure is in perfect agreement with the theoretical results and we can see that $\hat{\eta}$ grows and is positive when $s > \lambda$. In this case, as in the classical Euler, the entropy growth condition  is  equivalent to the Lax conditions.

%%%%%%%%%%%%%%%%%%%%%%%%%%%%%%%%%
%
%
%%%%%%%%%%%%%%%%%%%%%%%%%%%%%%%%%
\subsubsection{Monotonicity of the velocity} In this subsection, we discuss the monotonicity of the velocity along the shock curves.
\begin{proposition}\label{mono-velo} For the relativistic Euler system (\ref{main1}) with constitutive equations (\ref{press}), (\ref{ener-den}) and (\ref{num-den}), $\frac{dv}{dp}<0$ for the 1-shock curves, and $\frac{dv}{dp}>0$ for the 3-shock curves.
\end{proposition}
\begin{proof}
We only prove $\frac{dv}{dp}<0$ for the 1-shock curves since $\frac{dv}{dp}>0$ for the 3-shock curves can be proved in the same way.

As in (\ref{shock}), we can choose proper coordinate system such that $v_L=0$ to have
$$v^2=\frac{(p-p_L)(e-e_L)}{(p+e_L)(p_L+e)}c^2.$$
Differentiate the above equation with respect to $p$ to have

\begin{align*}
&\frac{dv^2}{dp}=2v\frac{dv}{dp}=\\
&\hspace{1cm}\frac{c^2(e_L+p_L)}{(p+e_L)^2(p_L+e)^2}\Big[(e-e_L)(e+p_L)+(p-p_L)(p+e_L)\frac{de}{dp}\Big].
\end{align*}
Note the fact $v<v_L=0$. To show $\frac{dv}{dp}<0$, it is equivalent to derive
\begin{equation}\label{ep1}
(e-e_L)(e+p_L)+(p-p_L)(p+e_L)\frac{de}{dp}>0.
\end{equation}
On the other hand, we use (\ref{ener-den}) and (\ref{curve}) to  have
$$\frac{de}{dp}=\gamma\frac{K_1(\gamma)}{K_2(\gamma)}+3+\Big[\gamma\Big(\frac{K_1(\gamma)}{K_2(\gamma)}\Big)^2
+4\frac{K_1(\gamma)}{K_2(\gamma)}-\gamma\Big]p\frac{d\gamma}{dp},$$
\begin{equation}\label{entr-grow2}
\begin{split}
&\frac{e_L+p_L}{n_L^2}\frac{dp}{d\gamma}=m^2c^4\frac{d}{d\gamma}\Big[\Big(\frac{1}{\gamma}\frac{K_1(\gamma)}{K_2(\gamma)}+
\frac{4}{\gamma^2}\Big)\Big(\gamma\frac{K_1(\gamma)}{K_2(\gamma)}+3+\frac{p_L}{p}\Big)\Big] =\\
&m^2c^4\Big\{\Big[\frac{1}{\gamma}\Big(\frac{K_1(\gamma)}{K_2(\gamma)}\Big)^2
+\frac{2}{\gamma^2}\frac{K_1(\gamma)}{K_2(\gamma)}-\frac{1}{\gamma}-\frac{8}{\gamma^3}\Big]
\Big(\gamma\frac{K_1(\gamma)}{K_2(\gamma)}+3+\frac{p_L}{p}\Big) +\\
&\Big(\frac{1}{\gamma}\frac{K_1(\gamma)}{K_2(\gamma)}+\frac{4}{\gamma^2}\Big)
\Big[\gamma\Big(\frac{K_1(\gamma)}{K_2(\gamma)}\Big)^2+4\frac{K_1(\gamma)}{K_2(\gamma)}-\gamma-\frac{p_L}{p}\frac{dp}{d\gamma}\Big]\Big\},
\end{split}
\end{equation}
and
\begin{equation}\label{entr-grow3}
\begin{split}
\frac{e_L+p_L}{n_L^2}=\frac{m^2c^4}{p}\Big(\frac{K_1(\gamma)}{\gamma K_2(\gamma)}+\frac{4}{\gamma^2}\Big)\frac{e+p_L}{e_L+p}
\end{split}
\end{equation}
We combine (\ref{entr-grow2}) and (\ref{entr-grow3}) to have
\begin{equation}\label{entr-grow4}
\begin{split}
\frac{1}{p}\frac{dp}{d\gamma}=&\frac{\Big[2\gamma^2\Big(\frac{K_1(\gamma)}{K_2(\gamma)}\Big)^3+13\gamma\Big(\frac{K_1(\gamma)}{K_2(\gamma)}\Big)^2+
(14-2\gamma^2)\frac{K_1(\gamma)}{K_2(\gamma)}-7\gamma-\frac{24}{\gamma}\Big]p
}
{\Big(\gamma\frac{K_1(\gamma)}{K_2(\gamma)}+4\Big)\Big(\frac{e+p_L}{e_L+p}p+p_L\Big)}+\\
&\frac{\Big[\gamma\Big(\frac{K_1(\gamma)}{K_2(\gamma)}\Big)^2+2\frac{K_1(\gamma)}{K_2(\gamma)}-\gamma-\frac{8}{\gamma}\Big]p_L}
{\Big(\gamma\frac{K_1(\gamma)}{K_2(\gamma)}+4\Big)\Big(\frac{e+p_L}{e_L+p}p+p_L\Big)},\\
\frac{de}{dp}=&\gamma\frac{K_1(\gamma)}{K_2(\gamma)}+3
+\frac{B_3(\gamma)\Big(\gamma\frac{K_1(\gamma)}{K_2(\gamma)}+4\Big)\Big(\frac{e+p_L}{e_L+p}p+p_L\Big)}
{\Big[2B_1(\gamma)-B_2(\gamma)\Big]p+B_2(\gamma)p_L}.
\end{split}
\end{equation}
Here and below, we use the following notations:
\begin{align*}
B_1(\gamma)=:&\gamma^2\Big(\frac{K_1(\gamma)}{K_2(\gamma)}\Big)^3+7\gamma\Big(\frac{K_1(\gamma)}{K_2(\gamma)}\Big)^2+
(8-\gamma^2)\frac{K_1(\gamma)}{K_2(\gamma)}-4\gamma-\frac{16}{\gamma},\\
B_2(\gamma)=:&\gamma\left(\frac{K_1(\gamma)}{K_2(\gamma)}\right)^2+2\frac{K_1(\gamma)}{K_2(\gamma)}-\gamma-\frac{8}{\gamma},\\
B_3(\gamma)=:&\gamma\Big(\frac{K_1(\gamma)}{K_2(\gamma)}\Big)^2+4\frac{K_1(\gamma)}{K_2(\gamma)}-\gamma.
\end{align*}
Note that from Appendix 3 Proposition \ref{eep}, one has
\begin{align}\label{B1231}
\begin{aligned}
B_1(\gamma)&=\Big[\gamma\left(\frac{K_1(\gamma)}{K_2(\gamma)}\right)^2
+3\frac{K_1(\gamma)}{K_2(\gamma)}- \gamma-\frac{4}{\gamma}\Big]\Big(\gamma\frac{K_1(\gamma)}{K_2(\gamma)}+4\Big)
<0,\\
B_2(\gamma)&=\Big[\gamma\left(\frac{K_1(\gamma)}{K_2(\gamma)}\right)^2
+3\frac{K_1(\gamma)}{K_2(\gamma)}- \gamma-\frac{3}{\gamma}\Big]-\frac{K_1(\gamma)}{K_2(\gamma)}- \frac{5}{\gamma}
<0,\\
B_3(\gamma)&=\frac{1}{(K_2(\gamma))^2}\Big[\gamma K^2_1(\gamma)+4K_1(\gamma)\Big(\frac{2K_1(\gamma)}{\gamma}+K_0(\gamma)\Big)- \\
&
\gamma\Big(\frac{2K_1(\gamma)}{\gamma}+K_0(\gamma)\Big)^2\Big] =\\
&\frac{1}{(K_2(\gamma))^2}\Big[\Big(\gamma+\frac{4}{\gamma}\Big) (K_1(\gamma))^2-\gamma (K_0(\gamma))^2\Big]
>0,
\end{aligned}
\end{align}
and
\begin{align}\label{B1232}
\begin{split}
&B_1(\gamma)-B_2(\gamma)=\\
&\gamma^2\Big(\frac{K_1(\gamma)}{K_2(\gamma)}\Big)^3+6\gamma\Big(\frac{K_1(\gamma)}{K_2(\gamma)}\Big)^2+
(6-\gamma^2)\frac{K_1(\gamma)}{K_2(\gamma)}-3\gamma-\frac{8}{\gamma}=\\
&\hspace{1.5cm}\gamma^2\left(\frac{K_1(\gamma)}{K_2(\gamma)}\right)^3+4\gamma\left(\frac{K_1(\gamma)}{K_2(\gamma)}\right)^2
-\gamma^2\frac{K_1(\gamma)}{K_2(\gamma)}-\gamma+\\
&\hspace{1.5cm}2\Big[\gamma\left(\frac{K_1(\gamma)}{K_2(\gamma)}\right)^2
+3\frac{K_1(\gamma)}{K_2(\gamma)}- \gamma-\frac{4}{\gamma}\Big]
<0,\\
B_1(\gamma)&-B_2(\gamma)-B_3(\gamma)=\\
& \Big[\gamma\left(\frac{K_1(\gamma)}{K_2(\gamma)}\right)^2
+3\frac{K_1(\gamma)}{K_2(\gamma)}- \gamma-\frac{4}{\gamma}\Big]
\Big(\gamma\frac{K_1(\gamma)}{K_2(\gamma)}+2\Big)
<0.
\end{split}
\end{align}
Now we use (\ref{ep1}) and (\ref{entr-grow4}) to have
\begin{equation*}
\begin{split}
 &(e-e_L)(e+p_L)+(p-p_L)(p+e_L)\times\\
 &\bigg[\gamma\frac{K_1(\gamma)}{K_2(\gamma)}+3
 +\frac{(\gamma\frac{K_1(\gamma)}{K_2(\gamma)}+4)B_3(\gamma)(\frac{e+p_L}{p+e_L}p+p_L)}{[2B_1(\gamma)-B_2(\gamma)]p+B_2(\gamma)p_L}\bigg]>0.
 \end{split}
 \end{equation*}
By (\ref{B1231}) and (\ref{B1232}), the equation above can be simplified as
\begin{equation}\label{vderi1}
\begin{split}
 &(p-p_L)(p+e_L)\Big\{\frac{e-e_L}{p-p_L}\frac{e+p_L}{e_L+p}B_2(\gamma)+B_3(\gamma)+\\
 &\Big(\gamma\frac{K_1(\gamma)}{K_2(\gamma)}+3\Big)\Big(B_2(\gamma)+B_3(\gamma)\Big)\Big\}p_L+\\
& (p-p_L)(e+p_L)\Big\{\Big(\frac{e-e_L}{p-p_L}+\frac{e_L+p}{e+p_L}\frac{e}{p}\Big)[2B_1(\gamma)-B_2(\gamma)]+\\
& \Big(\gamma\frac{K_1(\gamma)}{K_2(\gamma)}+4\Big)B_3(\gamma)\Big\}p<0.
 \end{split}
 \end{equation}
 Before verifying \eqref{vderi1}, we first show that
 \begin{equation}\label{ep11}
 \frac{e-e_L}{p-p_L}>1.
 \end{equation}
 In fact, we differentiate (\ref{entropias}) with respect to $p$ and use \eqref{S13} to have
\begin{equation}\label{entr-grow1}
\begin{split}
\frac{m}{k_B}\frac{dS}{dp}=&-\frac{1}{p}-\frac{2}{\gamma}\frac{d\gamma}{dp}
+\Big[\frac{K'_2(\gamma)}{K_2(\gamma)}+\frac{d}{d\gamma}\Big(\gamma\frac{K_1(\gamma)}{K_2(\gamma)}\Big)\Big]\frac{d\gamma}{dp}\\
=&-\frac{1}{p}+\Big[\gamma\Big(\frac{K_1(\gamma)}{K_2(\gamma)}\Big)^2
+3\frac{K_1(\gamma)}{K_2(\gamma)}-\gamma-\frac{4}{\gamma}\Big]\frac{d\gamma}{dp}>0.
\end{split}
\end{equation}
 Then we combine (\ref{entr-grow4}), (\ref{B1231}), (\ref{B1232}) and \eqref{entr-grow1} to have
\begin{equation*}
\begin{split}
B_1(\gamma)\Big(\frac{e+p_L}{p+e_L}-1\Big)p
<\left[B_1(\gamma)-B_2(\gamma)\right](p-p_L).
\end{split}
\end{equation*}
That is,
\begin{equation*}
\begin{split}
\frac{e-e_L}{p-p_L}-1>\frac{B_1(\gamma)-B_2(\gamma)}{B_1(\gamma)}\frac{p+e_L}{p}>0.
\end{split}
\end{equation*}
This inequality implies \eqref{ep11}. Now we use  (\ref{B1232}) and \eqref{ep11} to have
 \begin{equation*}
 \begin{split}
 &\frac{e-e_L}{p-p_L}\frac{e+p_L}{e_L+p}B_2(\gamma)+B_3(\gamma)<B_2(\gamma)+B_3(\gamma)<0,\\
 &\frac{e-e_L}{p-p_L}+\frac{e_L+p}{e+p_L}\frac{e}{p}=\frac{(ep-e_Lp_L)(e+p)}{(p-p_L)(e+p_L)p}=\\
& \hspace{0.5cm} \frac{e+p}{p}\Big(1+\frac{[(e-e_L)-(p-p_L)]p_L}{(p-p_L)(e+p_L)}\Big)>\frac{e+p}{p}=\gamma\frac{K_1(\gamma)}{K_2(\gamma)}+4.
 \end{split}
 \end{equation*}
 Then (\ref{vderi1}) holds since $2B_1(\gamma)-B_2(\gamma)+B_3(\gamma)=2[B_1(\gamma)-B_2(\gamma)]+[B_2(\gamma)+B_3(\gamma)]<0$ by (\ref{B1232}) and Appendix 3 (\ref{imp-ine1}).
\end{proof}

%%%%%%%%%%%%%%%%%%%%%%%%%%%%%%%%%
%
%
%%%%%%%%%%%%%%%%%%%%%%%%%%%%%%%%%
\subsection{Monotonicity of the velocity on rarefaction curves}
In this subsection, we consider the monotonicity of the velocity on rarefaction curves.
\begin{proposition}\label{mono-rare} For the relativistic Euler system (\ref{main1}) with constitutive equations (\ref{press}), (\ref{ener-den}) and (\ref{num-den}), $\frac{dv}{dp}<0$ on the 1-rarefaction curves, and $\frac{dv}{dp}>0$ on the 3-rarefaction curves.
\end{proposition}
\begin{proof}
Here we also only prove the case for 1-rarefaction curves, the other case for 3-rarefaction curves can be proved similarly.
From (\ref{1sp}), we have
\begin{equation*}
\frac{dv}{dp}=-\frac{\sqrt{(c^2-v^2)e_{p}}}{{(e+p)c}}<0.
\end{equation*}

\end{proof}

%%%%%%%%%%%%%%%%%%%%%
%
%
%%%%%%%%%%%%%%%%%%%%%%%
\section{Relativistic Euler system for Diatomic gas}

In this section, we analyze basic properties of the relativistic Euler system (\ref{main1}) in the case of gas  with internal structure (polyatomic gas). Contents in this section are almost along the same line of Section 3. Therefore, we will present the corresponding results in a brief way.

For the relativistic Euler system for polyatomic gas, the corresponding constitutive equations are given in \eqref{3n} and (\ref{ener-den}). We can rewrite them also in this equivalent form:
\begin{align}
&p=k_B nT=\frac{k_B}{m}\rho T,\label{gp-pre}\\
&e=\frac{nmc^2}{A(\gamma)K_2(\gamma)}\int_0^{\infty}\Big[3\frac{K_2(\gamma^*)}{\gamma^*}+K_1(\gamma^*)\Big]\phi(\mathcal {I})d\mathcal {I},\label{gp-ener}\\
&\displaystyle n=\frac{4\pi m^3c^3h^{-3}A(\gamma)K_2(\gamma)}{\gamma}e^{\frac{-S}{k_B}}\exp\Big\{\displaystyle \int_0^{\infty}\Big[3\frac{K_2(\gamma^*)}{\gamma^*}+K_1(\gamma^*)\Big]\phi(\mathcal {I})d\mathcal {I}\Big\},\label{gp-num}
\end{align}
where
$$A(\gamma)=\frac{\gamma}{K_2(\gamma)}\int_0^{\infty}\frac{K_2(\gamma^*)}{\gamma^*}{\mathcal {I}}^a d\mathcal {I}.$$
Then for $a=0$, we use  Appendix 1 (\ref{deriva}) to have
$$A(\gamma)=\frac{\gamma}{K_2(\gamma)}\int_0^{\infty}\frac{K_2(\gamma^*)}{\gamma^*}d\mathcal {I}=
\frac{\gamma}{K_2(\gamma)}\frac{mc^2}{\gamma}\frac{K_1(\gamma)}{\gamma}=\frac{mc^2K_1(\gamma)}{\gamma K_2(\gamma)}.$$
In the rest part of this paper, we focus on the case $a=0$ (diatomic gases). Namely, $\phi(\gamma)=1$ and
\begin{equation*}
\begin{split}
A(\gamma)&=\frac{mc^2K_1(\gamma)}{\gamma K_2(\gamma)}\\
\int_0^{\infty}\Big[3\frac{K_2(\gamma^*)}{\gamma^*}+K_1(\gamma^*)\Big]\phi(\mathcal {I})d\mathcal {I}&=\frac{mc^2}{\gamma}\Big[3\frac{K_1(\gamma)}{\gamma}+K_0(\gamma)\Big].
\end{split}
\end{equation*}
Then, in our case discussed below, the constitutive equations given in (\ref{gp-pre}), (\ref{gp-ener}) and (\ref{gp-num}) take the special form:
\begin{eqnarray}
&&p=k_B nT,\label{p-pre}\\
&&e=nk_BT\left(\frac{\gamma K_0(\gamma)}{K_1(\gamma)}+3\right)=p\left(\frac{\gamma K_0(\gamma)}{K_1(\gamma)}+3\right),\label{p-ener}\\
&&n=\frac{4\pi e^4 m^4c^5h^{-3}K_1(\gamma)}{\gamma^2}e^{\frac{\gamma K_0(\gamma)}{K_1(\gamma)}}e^{-\frac{S}{k_B}}.\label{p-num}
\end{eqnarray}

\begin{proposition} Under the relations (\ref{p-pre})-(\ref{p-num}), the system (\ref{main2}) is strictly hyperbolic and the corresponding characteristic velocities are sub-luminal.
%\begin{equation}\label{hyper00}
%\lambda_-=\lambda_1<v=\lambda_2<\lambda_+=\lambda_3.
%\end{equation}
\end{proposition}
\begin{proof}Corresponding to (\ref{cV}) and $(\ref{cV1})_2 $, we use (\ref{p-ener}) to have
 \begin{align*}
\begin{aligned}
\varepsilon=c^2\left[\frac{K_0 \left(\gamma
\right)}{K_1 \left(\gamma
\right)}+\frac{3}{\gamma}-1\right],
\end{aligned}
\end{align*}
and
\begin{align*}
\begin{aligned}
c_V=&\frac{d \varepsilon}{d T}=\rho c^2\frac{\partial}{\partial \gamma}\left[\frac{K_0 \left(\gamma
\right)}{K_1 \left(\gamma
\right)}+\frac{3}{\gamma}-1\right]\frac{d\gamma}{d T}=\\
&-\frac{k_B \rho}{  m}  \left[\gamma^2\left(\frac{K_0(\gamma)}{K_1(\gamma)}\right)^2+\gamma\frac{K_0(\gamma)}{K_1(\gamma)}- \gamma^2-3\right]>0,
\end{aligned}
\end{align*}
since for $\gamma\leq 3$,
\begin{equation}\label{p-cV1}
\gamma^2\left(\frac{K_0(\gamma)}{K_1(\gamma)}\right)^2+\gamma\frac{K_0(\gamma)}{K_1(\gamma)}-\gamma^2-3<\gamma-3\leq0,
\end{equation}
and for $\gamma>3$,
\begin{equation}\label{p-cV2}
\begin{split}
&\gamma^2\left(\frac{K_0(\gamma)}{K_1(\gamma)}\right)^2+\gamma\frac{K_0(\gamma)}{K_1(\gamma)}-\gamma^2-3 \leq\\
&\hspace{1cm} \gamma^2\left(1-\frac{1}{2\gamma}+\frac{3}{8\gamma^2}\right)^2+\gamma-\frac{1}{2}+\frac{3}{8\gamma}-\gamma^2-3=\\
&\hspace{1cm}-\frac{5}{2}+\frac{9}{64\gamma^2} <0,
\end{split}
\end{equation}
by Appendix 2 (\ref{acurate}). Then $c_p>0$ and the strict hyperbolicity follows from \eqref{case44}.

To prove the characteristic velocities of the system (\ref{main2}) with constitutive equations (\ref{p-pre})-(\ref{p-num}) are sub-luminal, as in the proof of Theorem \ref{sub-luminal}, we only need to prove that $e_p>3$. In fact, As in (\ref{pz}), under the relations (\ref{p-pre})-(\ref{p-num}), we have
\begin{align*}
\begin{aligned}
e_{p}=&\gamma\frac{K_0(\gamma)}{K_1(\gamma)}+3+\frac{\gamma\left(\frac{K_0(\gamma)}{K_1(\gamma)}\right)^2
+2\frac{K_0(\gamma)}{K_1(\gamma)}-\gamma}{\gamma\left(\frac{K_0(\gamma)}{K_1(\gamma)}\right)^2
+\frac{K_0(\gamma)}{K_1(\gamma)}-\gamma-\frac{4}{\gamma}} =\\
&3+\frac{\gamma^2\left(\frac{K_0(\gamma)}{K_1(\gamma)}\right)^3+2\gamma\left(\frac{K_0(\gamma)}{K_1(\gamma)}\right)^2-(\gamma^2
	+2)\frac{K_0(\gamma)}{K_1(\gamma)}- \gamma}{\gamma\left(\frac{K_0(\gamma)}{K_1(\gamma)}\right)^2
+\frac{K_0(\gamma)}{K_1(\gamma)}-\gamma-\frac{4}{\gamma}}>3,
\end{aligned}
\end{align*}
by (\ref{p-cV1}) (\ref{p-cV2}) and Appendix 3 (\ref{imp-inep}).
\end{proof}
The behavior of $\hat{\lambda}_3$ is plotted in Figure \ref{curves} ($a=0$).
 %to (\ref{cV}) and $(\ref{cV1})_2 $, we use (\ref{p-ener}) to have
% \begin{align}\label{p-cV1}
%\begin{aligned}
%\varepsilon=c^2\left[\frac{K_0 \left(\gamma
%\right)}{K_1 \left(\gamma
%\right)}+\frac{3}{\gamma}-1\right],
%\end{aligned}
%\end{align}
%and
%\begin{align*}
%\begin{aligned}
%c_V=&\frac{\partial \varepsilon}{\partial T}=c^2\frac{\partial}{\partial \gamma}\left[\frac{K_0 \left(\gamma
%\right)}{K_1 \left(\gamma
%\right)}+\frac{3}{\gamma}-1\right]\frac{\partial\gamma}{\partial T},\\
%=&-\frac{k_B}{  m}  \left[\gamma^2\left(\frac{K_0(\gamma)}{K_1(\gamma)}\right)^2+\gamma\frac{K_0(\gamma)}{K_1(\gamma)}- \gamma^2-3\right]>0,
%\end{aligned}
%\end{align*}
%since for $\gamma\leq 3$,
%\begin{equation}\label{p-cV1}
%\gamma^2\left(\frac{K_0(\gamma)}{K_1(\gamma)}\right)^2+\gamma\frac{K_0(\gamma)}{K_1(\gamma)}-\gamma^2-3<\gamma-3\leq0,
%\end{equation}
By the same arguments in Subsection \ref{sec-genui}, we can also show that the second eigenvalue of  the system (\ref{main2}) is linear degenerate under the relations (\ref{p-pre})-(\ref{p-num}). Moreover, we have the following proposition.

\begin{proposition} Under the relations (\ref{p-pre})-(\ref{p-num}), the inequality corresponding to (\ref{negativity}) holds. The first and third eigenvalues of the system (\ref{main2}) are genuinely nonlinear.
%\begin{equation}\label{hyper00}
%\lambda_-=\lambda_1<v=\lambda_2<\lambda_+=\lambda_3.
%\end{equation}
\end{proposition}
\begin{proof}
As in the proof of  Proposition \ref{genuine} in Appendix 4, we only need to show
\begin{align}\label{p-pp1}
\begin{aligned}
&(e+p)e_{pp}-2e_{p}(e_{p}-1)=\\
&e_{p}(-2e_{p}+3)+\frac{e}{p}\Big(e_{p}-\frac{e}{p}-1\Big)
+\frac{e+p}{\partial_\gamma p}\partial_\gamma\Big(\frac{p}{\partial_\gamma p}\frac{d}{d\gamma}\Big(\gamma\frac{K_0(\gamma)}{K_1(\gamma)}\Big)\Big)
<\\
&-9+\Big(\gamma\frac{K_0(\gamma)}{K_1(\gamma)}+3\Big)\frac{\frac{K_0(\gamma)}{K_1(\gamma)}+\frac{4}{\gamma}}
{\gamma\Big(\frac{K_0(\gamma)}{K_1(\gamma)}\Big)^2
+\frac{K_0(\gamma)}{K_1(\gamma)}-\gamma-\frac{4}{\gamma}}+\\
&\Big(\gamma\frac{K_0(\gamma)}{K_1(\gamma)}+4\Big)\Bigg[\frac{\frac{1}{\gamma}}{\gamma\Big(\frac{K_0(\gamma)}{K_1(\gamma)}\Big)^2
+\frac{K_0(\gamma)}{K_1(\gamma)}-\gamma-\frac{4}{\gamma}}-\Big(\frac{K_0(\gamma)}{K_1(\gamma)}+\frac{4}{\gamma}\Big)\times \\
&\frac{2\gamma\Big(\frac{K_0(\gamma)}{K_1(\gamma)}\Big)^3
+4\Big(\frac{K_0(\gamma)}{K_1(\gamma)}\Big)^2
+\Big(\frac{1}{\gamma}-2\gamma\Big)\frac{K_0(\gamma)}{K_1(\gamma)}-2+\frac{4}{\gamma^2}}
{\Big(\gamma\Big(\frac{K_0(\gamma)}{K_1(\gamma)}\Big)^2+3\frac{K_0(\gamma)}{K_1(\gamma)}-\gamma-\frac{4}{\gamma}\Big)^3}\Bigg]=\\
&-9+\frac{\Big(\gamma\frac{K_0(\gamma)}{K_1(\gamma)}+4\Big)\Big(\frac{K_0(\gamma)}{K_1(\gamma)}+\frac{4}{\gamma}\Big)}
{\Big(\gamma\Big(\frac{K_0(\gamma)}{K_1(\gamma)}\Big)^2+\frac{K_0(\gamma)}{K_1(\gamma)}-\gamma-\frac{4}{\gamma}\Big)^3}\times\mathcal {I}_3(\gamma)
\Big]<0,
\end{aligned}
\end{align}
where $$\mathcal {I}_3(\gamma)=\gamma^2\Big[1-\Big(\frac{K_0(\gamma)}{K_1(\gamma)}\Big)^2\Big]^2-11\Big(\frac{K_0(\gamma)}{K_1(\gamma)}\Big)^2
-\frac{9}{\gamma}\frac{K_0(\gamma)}{K_1(\gamma)}+10+\frac{12}{\gamma^2}.$$
We first use (\ref{p-cV1}) and (\ref{p-cV2}) to have
\begin{equation}\label{p-pp2}
\gamma\left(\frac{K_0(\gamma)}{K_1(\gamma)}\right)^2+\frac{K_0(\gamma)}{K_1(\gamma)}-\gamma-\frac{4}{\gamma}
<\gamma\left(\frac{K_0(\gamma)}{K_1(\gamma)}\right)^2+\frac{K_0(\gamma)}{K_1(\gamma)}-\gamma-\frac{3}{\gamma}<0.
\end{equation}
Then we combine (\ref{p-pp1}) and (\ref{p-pp2}) to see that in order to prove (\ref{negativity}), one only need to show
\begin{equation}\label{p-nega}
\mathcal {I}_3(\gamma)>0.
\end{equation}
Now we come to prove (\ref{p-nega}). We first note the fact
\begin{align*}
\begin{aligned}
\mathcal {I}_3(\gamma)\geq&-1 -\frac{9}{\gamma}+\frac{12}{\gamma^2}.\notag
 \end{aligned}
\end{align*}
This inequality implies that (\ref{p-nega}) holds for $\gamma\leq \gamma_1$, where $\gamma_1=\frac{-9+\sqrt{129}}{2}>1.1789>\gamma_0$ satisfies
$$\gamma^2_1+9\gamma_1-12=0.$$
Next we show that (\ref{negativity}) holds for $\gamma\in (\gamma_1, \sqrt{2}]$. In fact, we use Appendix 1 (\ref{new1}) in Proposition \ref{new-K01p} to have
\begin{align*}% \label{next-z1}
\mathcal {I}_3(\gamma&)>(\gamma_0-1)^2-11\Big(1-\frac{\gamma_0-1}{\gamma}\Big)^2
-\frac{9}{\gamma}\Big(1-\frac{\gamma_0-1}{\gamma}\Big)+10+\frac{12}{\gamma^2}>\\
&-11\Big[1-\frac{2\gamma_0-2}{\gamma}+\frac{(\gamma_0-1)^2}{\gamma^2}\Big]-\frac{9}{\gamma}+\frac{9(\gamma_0-1)}{\gamma^2}+10+\frac{12}{\gamma^2}>\\
&-1+\frac{22\gamma_0-31}{\gamma}+\frac{3+9\gamma_0-11(\gamma_0-1)^2}{\gamma^2}=\\
&-1-\frac{13}{2\gamma}+\frac{25}{2\gamma^2}>2,\hspace{3cm} \gamma\in (\gamma_1, \sqrt{2}].
 \end{align*}
Finally, we prove  (\ref{negativity}) holds for $\gamma\in (\sqrt{2}, \infty)$. We use Appendix 2 (\ref{rough}) to obtain
\begin{align*}% \label{next-\gamma1}
\mathcal {I}_3(\gamma)\geq&\gamma^2\Big(\frac{1}{2\gamma}-\frac{3}{8\gamma^2}-\frac{3}{16\gamma^3}\Big)^2\Big(2-\frac{1}{2\gamma}\Big)^2-\\
&11\Big(1-\frac{1}{2\gamma}+\frac{3}{8\gamma^2}+\frac{3}{16\gamma^3}\Big)^2-\\
&\frac{9}{\gamma}\Big(1-\frac{1}{2\gamma}+\frac{3}{8\gamma^2}+\frac{3}{16\gamma^3}\Big)+10+\frac{12}{\gamma^2}=\\
&\frac{1}{4}\Big(1-\frac{3}{4\gamma}-\frac{3}{8\gamma^2}\Big)^2\Big(4-\frac{2}{\gamma}+\frac{1}{4\gamma^2}\Big)
-\\
&11\Big(1-\frac{1}{\gamma}+\frac{1}{\gamma^2}-\frac{3}{64\gamma^4}+\frac{9}{64\gamma^5}+\frac{9}{256\gamma^6}\Big)\\
&+10-\frac{9}{\gamma}+\frac{33}{2\gamma^2}-\frac{27}{8\gamma^3}-\frac{27}{16\gamma^4}=\\
&\Big(1-\frac{1}{2\gamma}+\frac{1}{16\gamma^2}\Big)\Big(1-\frac{3}{2\gamma}-\frac{3}{16\gamma^2}+\frac{9}{16\gamma^3}+\frac{9}{64\gamma^4}\Big)\\
&-1+\frac{2}{\gamma}+\frac{11}{2\gamma^2}-\frac{27}{8\gamma^3}-\frac{75}{64\gamma^4}-\frac{99}{64\gamma^5}-\frac{99}{256\gamma^6}=\\
&\Big(\frac{3}{4}+\frac{-3+1}{16}+\frac{11}{2}\Big)\frac{1}{\gamma^2}+\Big(\frac{9}{16}-\frac{27}{8}\Big)\frac{1}{\gamma^3}\\
&+\Big(\frac{9}{64}-\frac{9}{32}-\frac{3}{256}-\frac{75}{64}\Big)\frac{1}{\gamma^4}\\
&+\Big(-\frac{9}{128}+\frac{9}{256}-\frac{99}{64}\Big)\frac{1}{\gamma^5}+\Big(\frac{9}{16\times64}-\frac{99}{256}\Big)\frac{1}{\gamma^6}=\\
&\frac{49}{8\gamma^2}-\frac{45}{16\gamma^3}-\frac{339}{256\gamma^4}-\frac{405}{256\gamma^5}-\frac{387}{1024\gamma^6}>0,
 \end{align*}
for $\gamma\in (\sqrt{2}, \infty)$.

\end{proof}

Corresponding to the monotonicity of the velocity in Proposition \ref{mono-velo}, we have the following proposition.

\begin{proposition} For the relativistic Euler system (\ref{main1}) with constitutive equations (\ref{p-pre})-(\ref{p-num}), $\frac{dv}{dp}<0$ for the 1-shock curves, and $\frac{dv}{dp}>0$ for the 3-shock curves.
\end{proposition}
\begin{proof}
As in Proposition \ref{mono-velo}, we choose proper coordinate system such that $v_L=0$. Similar to the derivation in Proposition \ref{mono-velo}, we can show that $\frac{dv}{dp}<0$ is equivalent to derive
\begin{equation}\label{p-ep1}
(e-e_L)(e+p_L)+(p-p_L)(p+e_L)\frac{de}{dp}>0.\qquad\mbox{and}
\end{equation}
\begin{equation}\label{pentr-grow4}
\begin{split}
\frac{de}{dp}=&\gamma\frac{K_0(\gamma)}{K_1(\gamma)}+3
+\frac{\bar{B}_3(\gamma)\left(\gamma\frac{K_0(\gamma)}{K_1(\gamma)}+4\right)\left(\frac{e+p_L}{e_L+p}p+p_L\right)}
{\left[2\bar{B}_1(\gamma)-\bar{B}_2(\gamma)\right]p+\bar{B}_2(\gamma)p_L}.
\end{split}
\end{equation}
Here and below, we use the following notations:
\begin{equation*}
\begin{split}
\bar{B}_1(\gamma)=:&\gamma^2\Big(\frac{K_0(\gamma)}{K_1(\gamma)}\Big)^3+5\gamma\Big(\frac{K_0(\gamma)}{K_1(\gamma)}\Big)^2
-\gamma^2\frac{K_0(\gamma)}{K_1(\gamma)}-4\gamma-\frac{16}{\gamma},\\
\bar{B}_2(\gamma)=:&\gamma\left(\frac{K_0(\gamma)}{K_1(\gamma)}\right)^2-\gamma-\frac{8}{\gamma},\\
\bar{B}_3(\gamma)=:&\gamma\Big(\frac{K_0(\gamma)}{K_1(\gamma)}\Big)^2+2\frac{K_0(\gamma)}{K_1(\gamma)}-\gamma.
\end{split}
\end{equation*}
We use Appendix Proposition \ref{new-K01p}, \ref{K01p}, \ref{peep} to have
\begin{align}\label{pB1231}
\begin{aligned}
\bar{B}_1(\gamma)&=\Big[\gamma\left(\frac{K_0(\gamma)}{K_1(\gamma)}\right)^2
+\frac{K_0(\gamma)}{K_1(\gamma)}- \gamma-\frac{4}{\gamma}\Big]\Big(\gamma\frac{K_0(\gamma)}{K_1(\gamma)}+4\Big)
<0,\\
\bar{B}_3(\gamma)&=\gamma\Big(\frac{K_0(\gamma)}{K_1(\gamma)}\Big)^2+2\frac{K_0(\gamma)}{K_1(\gamma)}-\gamma>0,
\end{aligned}
\end{align}
and
\begin{align}\label{pB1232}
\begin{aligned}
&\bar{B}_1(\gamma)-\bar{B}_2(\gamma)=\\
&\hspace{1cm}\gamma^2\Big(\frac{K_0(\gamma)}{K_1(\gamma)}\Big)^3+4\gamma\Big(\frac{K_0(\gamma)}{K_1(\gamma)}\Big)^2
-\gamma^2\frac{K_0(\gamma)}{K_1(\gamma)}-3\gamma-\frac{8}{\gamma}=\\
&\hspace{1cm}\Big[\gamma^2\left(\frac{K_0(\gamma)}{K_1(\gamma)}\right)^3+2\gamma\left(\frac{K_0(\gamma)}{K_1(\gamma)}\right)^2-(\gamma^2
	+2)\frac{K_0(\gamma)}{K_1(\gamma)}- \gamma\Big]+\\
&\hspace{1cm}2\Big[\gamma\left(\frac{K_0(\gamma)}{K_1(\gamma)}\right)^2
+\frac{K_0(\gamma)}{K_1(\gamma)}- \gamma-\frac{4}{\gamma}\Big]<0,\\
&\bar{B}_1(\gamma)-\bar{B}_2(\gamma)-\bar{B}_3(\gamma)=\\
&\hspace{1cm}\Big[\gamma\left(\frac{K_0(\gamma)}{K_1(\gamma)}\right)^2
+\frac{K_0(\gamma)}{K_1(\gamma)}- \gamma-\frac{4}{\gamma}\Big]\Big(\gamma\frac{K_0(\gamma)}{K_1(\gamma)}+2\Big)
<0.
\end{aligned}
\end{align}
Now we use (\ref{p-ep1}) and (\ref{pentr-grow4}) to have
\begin{equation*}
\begin{split}
 &(e-e_L)(e+p_L)+(p-p_L)(p+e_L)\\
 &\times\bigg[\gamma\frac{K_0(\gamma)}{K_1(\gamma)}+3
 +\frac{(\gamma\frac{K_0(\gamma)}{K_1(\gamma)}+4)\bar{B}_3(\gamma)(\frac{e+p_L}{p+e_L}p+p_L)}{[2\bar{B}_1(\gamma)-\bar{B}_2(\gamma)]p+\bar{B}_2(\gamma)p_L}\bigg]>0.
 \end{split}
 \end{equation*}
By (\ref{pB1231}) and (\ref{pB1232}), the equation above can be further simplified as
\begin{equation}\label{p-velo-deri1}
\begin{split}
 &(p-p_L)(p+e_L)\Big\{\frac{e-e_L}{p-p_L}\frac{e+p_L}{e_L+p}\bar{B}_2(\gamma)+\bar{B}_3(\gamma)+\\
 & \Big(\gamma\frac{K_0(\gamma)}{K_1(\gamma)}+3\Big)\Big(\bar{B}_2(\gamma)+\bar{B}_3(\gamma)\Big)\Big\}p_L+\\
& (p-p_L)(e+p_L)\Big\{\Big(\frac{e-e_L}{p-p_L}+\frac{e_L+p}{e+p_L}\frac{e}{p}\Big)[2\bar{B}_1(\gamma)-\bar{B}_2(\gamma)]+\\
& \Big(\gamma\frac{K_0(\gamma)}{K_1(\gamma)}+4\Big)\bar{B}_3(\gamma)\Big\}p<0.
 \end{split}
 \end{equation}
 By almost the same derivation of \eqref{ep1}
 \begin{equation*}
 \frac{e-e_L}{p-p_L}>1.
 \end{equation*}
Then we further use (\ref{pB1231}), (\ref{pB1232}) to have
 \begin{equation*}
 \begin{split}
 &\frac{e-e_L}{p-p_L}\frac{e+p_L}{e_L+p}\bar{B}_2(\gamma)+\bar{B}_3(\gamma)<\bar{B}_2(\gamma)+\bar{B}_3(\gamma)<0,\\
 &\frac{e-e_L}{p-p_L}+\frac{e_L+p}{e+p_L}\frac{e}{p}=\frac{(ep-e_Lp_L)(e+p)}{(p-p_L)(e+p_L)p}=\\
& \hspace{1cm} \frac{e+p}{p}\Big(1+\frac{[(e-e_L)-(p-p_L)]p_L}{(p-p_L)(e+p_L)}\Big)>\\
& \hspace{1cm}\frac{e+p}{p}=\gamma\frac{K_0(\gamma)}{K_1(\gamma)}+4.
 \end{split}
 \end{equation*}
 Then (\ref{p-velo-deri1}) holds since $2\bar{B}_1(\gamma)-\bar{B}_2(\gamma)+\bar{B}_3(\gamma)=2[\bar{B}_1(\gamma)-\bar{B}_2(\gamma)]+[\bar{B}_2(\gamma)+\bar{B}_3(\gamma)]<0$ by (\ref{pB1232}).
\end{proof}
The behavior of the production of entropy across the shock is given in Figure \ref{fig:etarel} with $a=0$.
By the same arguments as in Proposition \ref{mono-rare} for the monatomic gas, we also have the monotonicity of the velocity on rarefaction curves.
\begin{proposition} For the relativistic Euler system (\ref{main1}) with constitutive equations (\ref{p-pre})-(\ref{p-num}), $\frac{dv}{dp}<0$ on the 1-rarefaction curves, and $\frac{dv}{dp}>0$ on the 3-rarefaction curves.
\end{proposition}
%\begin{proof}
%Here we also only prove the case for 1-rarefaction curves, the other case for 3-rarefaction curves can be proved similarly.
%From (\ref{1sp}), we have
%\begin{equation*}\label{1sp}
%\frac{dv}{dp}=-\frac{\sqrt{(c^2-v^2)e_{p}}}{{(e+p)c}}<0.
%\end{equation*}
%
%
%
%
%\end{proof}

%%%%%%%%%%%%%%%%%%%%%
%
%
%%%%%%%%%%%%%%%%%%%%%%%
\section{Proof of the main theorem \ref{main theorem}}
In this section, we are devoted to the proof of the main Theorem \ref{main theorem}. We first discuss the condition under which vacuum occurs. Then for the case that no vacuum occurs,  based on the analysis about the structure of shock curves and the monotonicity of the velocity on rarefaction curves in Section 3 and 4, we solve the Riemann problem of the relativistic Euler system (\ref{main1}) with  constitutive equations given in (\ref{press}), (\ref{ener-den}) and (\ref{num-den}) or in  (\ref{p-pre}), (\ref{p-ener}) and (\ref{p-num}).
%%%%%%%%%%%%%%%%%%%%%%%%%%%%%%%%%
%
%
%%%%%%%%%%%%%%%%%%%%%%%%%%%%%%%%%
\subsection{Vacuum condition} To give a vacuum condition, we first define vacuum. We say that vacuum occurs if
 $$e=0.$$
In fact, we have obtained that  the pressure $p$ is monotonic along 1-curves and 3-curves, and is constant along 2-contact discontinuity waves. Note that $e=p\left(\gamma\frac{K_1(\gamma)}{K_2(\gamma)}+3\right)$ for monatomic gas and $e=p\left(\gamma\frac{K_0(\gamma)}{K_1(\gamma)}+3\right)$ for diatomic gas. It is easy to
see that the vacuum may occur only when a 1-rarefaction wave interacts a 3-rarefaction wave.

Denote $(\bar{r}_L, \bar{s}_L)$ and $(\bar{r}_R, \bar{s}_R)$ as the Riemann invariants at $x < 0$ and $x>0$, respectively.
$\bar{r}_L, \bar{s}_L, \bar{r}_R$ and $\bar{s}_R$ can be written as
  \begin{align*}
   \bar{r}_L=&\frac12\ln\Big(\frac{c+v_L}{c-v_L}\Big)+\int_0^{p_L}\frac{\sqrt{e_{p}}dp}{(e+p)c}
  =\\
  &\frac12\ln\Big(\frac{c+v_L}{c-v_L}\Big)+\int_0^{e_L}\frac{de}{(e+p)\sqrt{e_{p}}c},\\
  \bar{s}_L=&\frac12\ln\Big(\frac{c+v_L}{c-v_L}\Big)-\int_0^{p_L}\frac{\sqrt{e_{p}}dp}{(e+p)c}
  =\\
  &\frac12\ln\Big(\frac{c+v_L}{c-v_L}\Big)-\int_0^{e_L}\frac{de}{(e+p)\sqrt{e_{p}}c},\\
  \bar{r}_R=&\frac12\ln\Big(\frac{c+v_R}{c-v_R}\Big)+\int_0^{p_R}\frac{\sqrt{e_{p}}dp}{(e+p)c}
 =\\
  &\frac12\ln\Big(\frac{c+v_R}{c-v_R}\Big)+\int_0^{e_R}\frac{de}{(e+p)\sqrt{e_{p}}c},\\
  \bar{s}_R=&\frac12\ln\Big(\frac{c+v_R}{c-v_R}\Big)-\int_0^{p_R}\frac{\sqrt{e_{p}}dp}{(e+p)c}
  =\\
  &\frac12\ln\Big(\frac{c+v_R}{c-v_R}\Big)-\int_0^{e_R}\frac{de}{(e+p)\sqrt{e_{p}}c}.
    \end{align*}
  Note that $\bar{r}$ is a constant along a 1-rarefaction curve, velocity $v$ and pressure $p$ are constant along a 2-contact discontinuity wave, and $\bar{s}$ is a constant along a 3-rarefaction curve. Therefore, there exist states $(e_1, v_1, S_1)$ and $(e_3, v_3, S_3)$ with $v_1=v_2, S_1=S_L, S_3=S_R$ such that
 \begin{align*}
  \begin{aligned}
  \bar{r}_L-\bar{s}_R=&\frac12\ln\Big(\frac{c+v_L}{c-v_L}\Big)+\int_0^{e_L}\frac{de}{(e+p)\sqrt{e_{p}}c}-\\
 & \Big[\frac12\ln\Big(\frac{c+v_R}{c-v_R}\Big)-\int_0^{e_R}\frac{de}{(e+p)\sqrt{e_{p}}c}\Big]=\\
  &\frac12\ln\Big(\frac{c+v_1}{c-v_1}\Big)+\int_0^{e_1}\frac{de}{(e+p)\sqrt{e_{p}}c}-\\
  &\Big[\frac12\ln\Big(\frac{c+v_3}{c-v_3}\Big)-\int_0^{e_3}\frac{de}{(e+p)\sqrt{e_{p}}c}\Big]=\\
  & \int_0^{e_1}\frac{de}{(e+p)\sqrt{e_{p}}c}+\int_0^{e_3}\frac{de}{(e+p)\sqrt{e_{p}}c}.
  \end{aligned}
  \end{align*}
Then $\bar{r}_L\leq\bar{s}_R$ implies that
$$\int_0^{e_1}\frac{de}{(e+p)\sqrt{e_{p}}c}+\int_0^{e_3}\frac{de}{(e+p)\sqrt{e_{p}}c}\leq0.$$
Namely, $e_1, e_3\leq0$ and vacuum occurs. Therefore, the condition that vacuum occurs is $\bar{r}_L\leq\bar{s}_R$.

%%%%%%%%%%%%%%%%%%%%%%%%%%%%%%%%%
%
%
%%%%%%%%%%%%%%%%%%%%%%%%%%%%%%%%%
\subsection{Existence of solutions to the Riemann problem}
Finally, we discuss the solutions to the Riemann problem for the case $\bar{r}_L>\bar{s}_R$. In Section 3 and Section 4, we have proved that $\frac{dv}{dp}<0$ on $\mathcal{T}_1^p(\mathbf{u}_L)$ and $\frac{dv}{dp}>0$ on $\mathcal{T}_3^p(\mathbf{u}_L)$ for the monatomic gas case and diatomic gas case, respectively. Denote
$$v=f_1(p; p_L, v_L, S_L) \quad\mbox{and}\quad v=f_3(p; p_R, v_R, S_R) $$
as the curves $\mathcal{T}_1^p(\mathbf{u}_L)$ and the backward 3-curve, respectively.  Then the $p-v$ plane is divided into four parts by the curves $\mathcal{T}_1^p(\mathbf{u}_L)$ and $\mathcal{T}_3^p(\mathbf{u}_L)$. Moreover,  if $(p_L, v_L)$ is above the curve $f_3(p; p_R, v_R, S_R)$, i.e.,
 $$v_L>f_3(p_L; p_R, v_R, S_R),$$
 the 1-curve of the Riemann problem should be a rarefaction curve; while if $(p_L, v_L)$ is below the curve $f_3(p; p_R, v_R, S_R)$,
  i.e.,
 $$v_L<f_3(p_L; p_R, v_R, S_R),$$ the 1-curve of the Riemann problem should be a shock curve.

 Correspondingly,
 the 3-curve of the Riemann problem should be a rarefaction curve if $(p_R, v_R)$ is above the curve $\mathcal{T}_1^p(\mathbf{u}_L)$, i.e.,
 $$v_R>f_1(p_R; p_L, v_L, S_L),$$  and the 3-curve of the Riemann problem should be a rarefaction curve if $(p_R, v_R)$ is below the curve $\mathcal{T}_1^p(\mathbf{u}_L)$, i.e.,
 $$v_R<f_1(p_R; p_L, v_L, S_L).$$

 Since $\frac{dv}{dp}<0$ on $\mathcal{T}_1^p(\mathbf{u}_L)$ and $\frac{dv}{dp}>0$ on 3-curve, an  intermediate state $(p_M, v_M)$ can be uniquely solved by
  $$v=f_1(p; p_L, v_L, S_L), \qquad v=f_3(p; p_R, v_R, S_R).$$
  %Namely, as illustrated in Figure \ref{curves}, the $p-v$ plane is divided into the following four parts by the curves $\mathcal{T}_1^p(\mathbf{u}_L)$ and $\mathcal{T}_3^p(\mathbf{u}_L)$: if $(p_R, v_R)$ is in region $\uppercase\expandafter{\romannumeral 1}$, $(p_L, v_L)$ is connected to an  intermediate state $(p_M, v_M)$ by a 1-shock curve first and  $(p_M, v_M)$ is further connected to  $(p_R, v_R)$ by a 3-rarefaction curve; if $(p_R, v_R)$ is in region $\uppercase\expandafter{\romannumeral 2}$, $(p_L, v_L)$ is connected to an  intermediate state $(p_M, v_M)$ by a 1-shock curve first and  $(p_M, v_M)$ is further connected to  $(p_R, v_R)$ by a 3-shock curve; if $(p_R, v_R)$ is in region $\uppercase\expandafter{\romannumeral 3}$, $(p_L, v_L)$ is connected to an  intermediate state $(p_M, v_M)$ by a 1-rarefaction curve first and  $(p_M, v_M)$ is further connected to  $(p_R, v_R)$ by a 3-shock curve; if $(p_R, v_R)$ is in region $\uppercase\expandafter{\romannumeral 4}$, $(p_L, v_L)$ is connected to an  intermediate state $(p_M, v_M)$ by a 1-rarefaction curve first and  $(p_M, v_M)$ is further connected to  $(p_R, v_R)$ by a 3-rarefaction curve.

Having obtained the 1-curve and 3-curve, we need to determine the 2-contact discontinuity wave to fully solve the Riemann problem. In fact, since the velocity $v$ and pressure $p$ are constants on the 2-contact discontinuity wave, we only need to obtain the entropy. Note that the entropy is a constant along a rarefaction curve and is monotonic along a shock curve. Then we can uniquely determine the entropy for the left state and right state of the 2-contact discontinuity wave by the values of $p$ on the 1-curve and 3-curve, respectively.

\section{Conclusions}
In this paper, we consider the relativistic Euler equation in one space dimension. The constitutive equations for the closure of the differential system come from the relativistic Boltzmann-Chernikov equation that involves the Synge energy in the case of a monatomic gas and the generalized Synge energy in the case of a polyatomic gas. These constitutive equations are more appropriate than the ones present in literatures to study the Riemann problem. In fact, using the Synge equation we discover the constitutive equations of previous papers listed in the introduction are valid only in the classical limit ($\gamma \rightarrow \infty$) or in the ultra-relativistic limit  ($\gamma \rightarrow 0$).
Therefore our analysis is more realistic in the relativistic regime ($\gamma$ small). This more physical case is mathematically difficult because the modified Bessel functions of the second kind appear in the constitutive equations. Nevertheless, we are able to prove rigorously the well-posedness of the Riemann problem at least for monatomic and diatomic gases. For these kinds of gases, our results reduce to one of the previous studies as limit cases of the classical regime or ultra-relativistic framework.

%%%%%%%%%%%%%%%%%%%%%
%
%
%%%%%%%%%%%%%%%%%%%%%%%
% %\section{Appendix}Therefore aour anaysys

\setcounter{equation}{0}
%In this section, we give definitions and basic properties of the modified Bessel functions. Several estimates with respect to these functions are provided as well. Moreover, these estimates imply the two conjectures in \cite{Speck-Strain-CMP-2011}, which are associated with the equivalence of the fluid variables, the hyperbolicity  of relativistic Euler system and the range of speed of sound in the relativistic setting. The estimates presented in this part are not only essential to the analysis in our paper but also probably interesting in terms of the functions themselves.

%%%%%%%%%%%%%%%%%%%%%%%%%%%%%%%%%
%
%
%%%%%%%%%%%%%%%%%%%%%%%%%%%%%%%%%
\section*{Appendix 1: Modified Bessel functions and properties}
In this part, we recall expressions of the modified Bessel functions and their basic properties. Moreover, with our observation, a simple corollary is also presented. Now we give the modified Bessel functions and collect their basic properties:
\begin{lemma}\cite{Groot-Leeuwen-Weert-1980,Oliver-1974,Wa}\label{def-pro}
	Let
	$K_j(\gamma)$ be the Bessel functions defined by
	\begin{equation}\label{defini}
	 K_j(\gamma)=\frac{(2^j)j!}{(2j)!}\frac{1}{\gamma^j}\int_{\lambda=\gamma}^{\lambda=\infty}e^{-\lambda}(\lambda^2-\gamma^2)^{j-1/2}d\lambda,\quad(j\geq0).
	\end{equation}
	Then the following identities hold:
	\begin{eqnarray}
	 &K_j(\gamma)=\frac{2^{j-1}(j-1)!}{(2j-2)!}\frac{1}{\gamma^j}\int_{\lambda=\gamma}^{\lambda=\infty}e^{-\lambda}(\lambda^2-\gamma^2)^{j-3/2}d\lambda,\quad(j>0),\nonumber\\
	&K_{j+1}(\gamma)=2j\frac{K_j(\gamma)}{\gamma}+K_{j-1}(\gamma),\quad(j\geq1), \label{transform}\\
	&K_j(\gamma)<K_{j+1}(\gamma),\quad (j\geq0), \notag
	\end{eqnarray}
	and
	\begin{eqnarray}
	&\frac{d}{d\gamma}\left(\frac{K_j(\gamma)}{\gamma^j}\right)=-\frac{K_{j+1}(\gamma)}{\gamma^j},\quad(j\geq0),\label{deriva}\\
	&\displaystyle K_{j}(\gamma)=\sqrt{\frac{\pi}{2\gamma}}e^{-\gamma}\left(\gamma_{j,n}(\gamma)\gamma^{-n}+\sum_{m=0}^{n-1}A_{j,m}\gamma^{-m}\right),\quad(j\geq0,~n\geq1),\label{remainder}
	\end{eqnarray}
	where expressions of the coefficients in (\ref{remainder}) are
	
	\begin{align} \label{coefficient}
	\begin{aligned}
	&A_{j,0}=1\\
	&A_{j,m}=\frac{(4j^2-1)(4j^2-3^2)\cdots(4j^2-(2m-1)^2)}{m!8^m},\quad(j\geq0,~m\geq1), \\
	&|\gamma_{j,n}(\gamma)|\leq2e^{[j^2-1/4]\gamma^{-1}}|A_{j,n}|,\quad(j\geq0,~n\geq1).
	\end{aligned}
	\end{align}
	
	On the other hand, according to \cite{Wa} (in Page 80), the Bessel functions defined in (\ref{defini}) can also be written in the following form:
	\begin{align}
		K_0(\gamma)=&-\sum^{\infty}_{m=0}\frac{(\frac{1}{2}\gamma)^{2m}}{m!m!}\Big[\ln\Big(\frac{\gamma}{2}\Big)-\psi(m+1)\Big],\nonumber\\
	K_n(\gamma)=&\frac{1}{2}\sum^{n-1}_{m=0}(-1)^m\frac{(n-m-1)!}{m!}\Big(\frac{1}{2}\gamma\Big)^{-n+2m}+\nonumber\\
	&(-1)^{n+1}\sum^{\infty}_{m=0}\frac{(\frac{1}{2}\gamma)^{n+2m}}{m!(m+n)!}\times\nonumber\\
&\Big[\ln\Big(\frac{\gamma}{2}\Big)
	-\frac{1}{2}\psi(n+m)-\frac{1}{2}\psi(n+m+1)\Big], \label{new-defi}\\
	\psi(1)=&-C_E,\quad \psi(m+1)=-C_E+\sum^{m}_{k=1}\frac{1}{k},\quad m\geq1,\nonumber\\
	 K_1(\gamma)=&\frac{1}{\gamma}+\sum^{\infty}_{m=0}\frac{(\frac{1}{2}\gamma)^{2m+1}}{m!(m+1)!}
\Big[\ln\Big(\frac{\gamma}{2}\Big)-\frac{1}{2}\psi(m+1)-\frac{1}{2}\psi(m+2)\Big],\nonumber
		\end{align}
	where $C_E=0.5772157\ldots$ is the Euler's constant.
	
\end{lemma}

From Lemma \ref{def-pro}, we immediately have the following corollary:

\begin{corollary}\label{K012} For $K_j(j\geq0)$ defined in Lemma \ref{def-pro}, it holds that
	\begin{equation}\label{K012-1}
	K_1^2\leq3K_0K_2,
	\end{equation}
	\begin{equation}\label{K012-2}
	3\left(\frac{K_0(\gamma)}{K_1(\gamma)}\right)^2+\frac{6}{\gamma}\frac{K_0(\gamma)}{K_1(\gamma)}-1\geq0.
	\end{equation}
\end{corollary}

\begin{proof} From (\ref{defini}), we have
	\begin{eqnarray*}
		&K_0(\gamma)=\int_{\lambda=\gamma}^{\lambda=\infty}e^{-\lambda}(\lambda^2-\gamma^2)^{-1/2}d\lambda,\\
		&K_1(\gamma)=\frac{1}{\gamma}\int_{\lambda=\gamma}^{\lambda=\infty}e^{-\lambda}(\lambda^2-\gamma^2)^{1/2}d\lambda,\\
		&K_2(\gamma)=\frac{1}{3\gamma^2}\int_{\lambda=\gamma}^{\lambda=\infty}e^{-\lambda}(\lambda^2-\gamma^2)^{3/2}d\lambda.
	\end{eqnarray*}
	These equations imply (\ref{K012-1}) by H$\ddot{o}$lder's inequality.
	
	Taking $j=1$ in (\ref{transform}) and inserting it to (\ref{K012-1}) yield
	$$K_1^2(\gamma)\leq 3K_0(\gamma)\left(\frac{2}{\gamma}K_1(\gamma)+K_0(\gamma)\right).$$
	Then (\ref{K012-2}) follows.
\end{proof}

%%%%%%%%%%%%%%%%%%%%%%%%%%%%%%%%%
%
%
%%%%%%%%%%%%%%%%%%%%%%%%%%%%%%%%%
\section*{Appendix 2:  Estimates of the ratio $\frac{K_0(\gamma)}{K_1(\gamma)}$}
In this subsection, we concentrate on estimates of $\frac{K_0(\gamma)}{K_1(\gamma)}$.
Our estimates are divided into two cases according to different expressions of $K_m (m\geq1)$ given in Lemma \ref{def-pro}: the case $\gamma\in (0, \sqrt{2}]$ by (\ref{new-defi}) and the case $\gamma\in [1.1, \infty)$ by (\ref{defini}). Here $\gamma_0=1.1229189\ldots$ is a constant satisfying
$$\ln\Big(\frac{\gamma}{2}\Big)+C_E=0.$$
We first estimate $\frac{K_0(\gamma)}{K_1(\gamma)}$ for the first case $\gamma\in (0, \sqrt{2}]$ by using the expressions (\ref{new-defi}).

\begin{proposition}\label{new-K01p} For $\gamma\in[\gamma_0, \sqrt{2}]$, it holds that
	\begin{equation}\label{new1}
	\frac{K_0(\gamma)}{K_1(\gamma)}\leq 1-\frac{\gamma_0-1}{\gamma}.
	\end{equation}
	And for $\gamma\in(0, \gamma_0]$, we have $\Big(\frac{K_0(\gamma)}{K_1(\gamma)}\Big)^2+\frac{2}{\gamma}\frac{K_0(\gamma)}{K_1(\gamma)}-1>0$. Moreover, we have
	\begin{equation}\label{low1}
	\frac{\gamma}{\sqrt{\gamma^2+1}+1}\leq\frac{K_0(\gamma)}{K_1(\gamma)}\leq \gamma\left[\frac{11}{16}-\left(\ln(\frac{\gamma}{2})+C_E\right)\right].
	\end{equation}
\end{proposition}
\begin{proof}
	We first prove (\ref{new1}). From (\ref{new-defi}), we get for $\gamma\in[\gamma_0, \sqrt{2}]$ that
	\begin{equation*}
	\gamma\frac{K_0(\gamma)}{K_1(\gamma)}\leq \frac{-[\ln(\frac{\gamma}{2})+C_E]\Big[\gamma^2+\frac{\gamma^4}{4}+\frac{\gamma^6e^{\frac{\gamma^2}{28}}}{64}\Big]+\frac{\gamma^4}{4}+
		-\frac{3\gamma^6 e^{\frac{\gamma^2}{28}}}{128}}
	{[\ln(\frac{\gamma}{2})+C_E]\Big[\frac{\gamma^2}{2}+\frac{\gamma^4}{16}+\frac{\gamma^6e^{\frac{\gamma^2}{28}}}{32\times12}\Big]+1-\frac{\gamma^2}{4}-
\frac{5\gamma^4}{64}+
		-\frac{5\gamma^6 e^{\frac{\gamma^2}{40}}}{32\times36}}.
	\end{equation*}
	Then, (\ref{new1}) holds if we can show
	\begin{equation*}
	\begin{split}
	& [\ln(\frac{\gamma}{2})+C_E]\Big[\gamma^2+\frac{\gamma^4}{4}+\frac{\gamma^6e^{\frac{\gamma^2}{28}}}{64}\Big]+\frac{\gamma^4}{4}+
		-\frac{3\gamma^6 e^{\frac{\gamma^2}{28}}}{128}\leq\\	 &\Big\{[\ln(\frac{\gamma}{2})+C_E]\Big[\frac{\gamma^2}{2}+\frac{\gamma^4}{16}+\frac{\gamma^6e^{\frac{\gamma^2}{28}}}{32\times12}\Big]+1-\frac{\gamma^2}{4}-
\frac{5\gamma^4}{64}+
		-\frac{5\gamma^6 e^{\frac{\gamma^2}{40}}}{32\times36}\Big\}\times\\
&(1+\gamma-\gamma_0).
	\end{split}
	\end{equation*}
	Namely,
	\begin{align}\label{new2}
	\begin{aligned}
	f(\gamma):=&1+\gamma-\gamma_0-\frac{1}{4}(1+\gamma-\gamma_0)\gamma^2-\frac{\gamma^4}{64}[5(\gamma-\gamma_0)+21]-\\
&\Big[5(1+\gamma-\gamma_0)e^{\frac{\gamma^2}{40}}+27e^{\frac{\gamma^2}{28}}\Big]\frac{\gamma^6}{32\times36}+\\
&\Big[\ln(\frac{\gamma}{2})+C_E\Big]
	\Big\{\frac{\gamma^2}{2}(3+\gamma-\gamma_0)+\frac{\gamma^4}{16}(5+\gamma-\gamma_0)+\\
	&\frac{\gamma^6}{32\times12}\Big[(1+\gamma-\gamma_0)e^{\frac{\gamma^2}{40}}+6e^{\frac{\gamma^2}{28}}\Big]\Big\}\geq0
	.
	\end{aligned}
	\end{align}
	
	Note that for $\gamma\in [\gamma_0, \sqrt{2}]$,
	\begin{align*}
		f'(\gamma)=&1-\frac{1}{4}\gamma^2-\frac{\gamma}{2}(1+\gamma-\gamma_0)-\frac{\gamma^3}{16}[5(\gamma-\gamma_0)+21]
	-\frac{5\gamma^4}{64}\\
&-\frac{\gamma^5}{32\times6}\Big[5(1+\gamma-\gamma_0)e^{\frac{\gamma^2}{40}}+27e^{\frac{\gamma^2}{28}}\Big]\\
	 &+\frac{\gamma}{2}(3+\gamma-\gamma_0)+\frac{\gamma^3}{16}(5+\gamma-\gamma_0)+\\
&\frac{\gamma^5}{32\times12}\Big[(1+\gamma-\gamma_0)e^{\frac{\gamma^2}{40}}+6e^{\frac{\gamma^2}{28}}\Big]-\\
	&\frac{\gamma^6}{32\times36}\Big\{5\Big[1+\frac{\gamma(1+\gamma-\gamma_0)}{20}\Big]e^{\frac{\gamma^2}{40}}+\frac{27\gamma}{14}e^{\frac{\gamma^2}{28}}\Big]\Big\}\\
	&+\Big[\ln(\frac{\gamma}{2})+C_E\Big]
	\Big\{(3+\gamma-\gamma_0)\gamma+\frac{\gamma^2}{2}+\frac{\gamma^3}{4}(5+\gamma-\gamma_0)+\frac{\gamma^4}{16}+\\
	&\frac{\gamma^5}{64}\Big[(1+\gamma-\gamma_0)e^{\frac{\gamma^2}{40}}+6e^{\frac{\gamma^2}{28}}\Big]\\
&	 +\frac{\gamma^6}{32\times12}\Big[\Big(1+\frac{\gamma(1+\gamma-\gamma_0)}{20}\Big)e^{\frac{\gamma^2}{40}}
+\frac{3\gamma}{7}e^{\frac{\gamma^2}{28}}\Big]\Big\}.
	\end{align*}
We can further obtain that for $\gamma\in [\gamma_0, \sqrt{2}]$,
\begin{equation*}
	\begin{split}
f''(\gamma)=&4-\gamma_0+\gamma+\frac{(2\gamma_0-7)\gamma^2}{4}-\gamma^3-\\
&\frac{\gamma^4}{128}\Big[13(1+\gamma-\gamma_0)e^{\frac{\gamma^2}{40}}
+68e^{\frac{\gamma^2}{28}}\Big]-\\
	&\frac{\gamma^6}{64}\Big\{\Big[3+\frac{3\gamma(1+\gamma-\gamma_0)}{20}\Big]e^{\frac{\gamma^2}{40}}
	+\frac{27\gamma}{7}e^{\frac{\gamma^2}{28}}\Big\}-\\
	&\frac{\gamma^6}{32\times36}\Big\{
	\Big[\frac{1-\gamma_0+3\gamma}{4}+\frac{\gamma^2(1+\gamma-\gamma_0)}{80}\Big]e^{\frac{\gamma^2}{40}}\\
	&+\Big(\frac{27}{14}+\frac{27\gamma^2}{256}\Big)e^{\frac{\gamma^2}{28}}\Big\}+\Big[\ln(\frac{\gamma}{2})+C_E\Big] \Big\{3-\gamma_0+3\gamma+\\
	&+\frac{(15-3\gamma_0)\gamma^2}{4}+\gamma^3
	\frac{5\gamma^4}{64}\Big[(1+\gamma-\gamma_0)e^{\frac{\gamma^2}{40}}+6e^{\frac{\gamma^2}{28}}\Big]+\\	 &\frac{\gamma^5}{32}\Big[\Big(1+\frac{\gamma(1+\gamma-\gamma_0)}{20}\Big)e^{\frac{\gamma^2}{40}}
+\frac{3\gamma}{7}e^{\frac{\gamma^2}{28}}\Big]+\frac{\gamma^6}{32\times12}\times\\
&\Big\{
\Big[\frac{1-\gamma_0+3\gamma}{20}+\frac{\gamma^2(1+\gamma-\gamma_0)}{400}\Big]e^{\frac{\gamma^2}{40}}
+\\
&\Big(\frac{3}{7}+\frac{3\gamma^2}{98}\Big)e^{\frac{\gamma^2}{28}}\Big\}\Big\}
	< 0.
	\end{split}
	\end{equation*}
Then we have
$$f(\gamma)\geq \min\{f(\gamma_0), f(\sqrt{2})\}$$ for $\gamma\in [\gamma_0, \sqrt{2}]$. On the other hand, we have
	\begin{align*}
f(\gamma_0)=&1-\frac{\gamma_0^2}{4}-\frac{21\gamma_0^4}{64}-\frac{\gamma_0^6}{32\times36}\Big[5e^{\frac{\gamma_0^2}{40}}+27e^{\frac{\gamma_0^2}{28}}\Big]>0,\quad\mbox{and}\\
	  f(\sqrt{2})=&1+\sqrt{2}-\gamma_0-\frac{1+\sqrt{2}-\gamma_0}{2}-\frac{5(\sqrt{2}-\gamma_0)+21}{16}-\\
	  &\Big[\dfrac{5(1+\sqrt{2}-\gamma_0)}{144}e^{\frac{1}{20}}+\dfrac{27}{144}e^{\frac{1}{14}}\Big]+\Big[\ln(\frac{\sqrt{2}}{2})+C_E\Big]\times
	  \\
	  &\Big\{3+\sqrt{2}-\gamma_0+\frac{5+\sqrt{2}-\gamma_0}{4}+\\
	  &\frac{1}{48}\Big[(1+\sqrt{2}-\gamma_0)e^{\frac{1}{20}}+6e^{\frac{1}{14}}\Big]\Big\}>0.
		\end{align*}
	Then (\ref{new2}) holds.
	
	Now we turn to the proof of (\ref{low1}). We first verify the left inequality of (\ref{low1}).  For $\gamma\in(0,\gamma_0],$ we use (\ref{new-defi}) to have
	\begin{equation*}
	\frac{K_0(\gamma)}{K_1(\gamma)}\geq \frac{-\left[\ln(\frac{\gamma}{2})+C_E\right]\gamma-\frac{1}{4}\left[\ln(\frac{\gamma}{2})+C_E-1\right]\gamma^3
	}
	{1+\frac{1}{2}\left[\ln(\frac{\gamma}{2})+C_E-\frac{1}{2}\right]\gamma^2+\frac{1}{16}\left[\ln(\frac{\gamma}{2})+C_E-\frac{5}{4}\right]\gamma^4
	}.
	\end{equation*}
	(\ref{low1}) holds if we have
	\begin{equation*}
	\begin{split}
	 \overline{f}(\gamma)=&-\Big[\ln(\frac{\gamma}{2})+C_E\Big]\Big[1+\sqrt{\gamma^2+1}+\frac{\gamma^2(3+\sqrt{\gamma^2+1})}{4}+\frac{\gamma^4}{16}\Big]+\\
	&\frac{\gamma^2(2+\sqrt{\gamma^2+1})}{4}+\frac{5\gamma^4}{64}-1
	>0.
	\end{split}
	\end{equation*}
	In fact, for $\gamma\in (0,\gamma_0)$, we have
	\begin{equation*}
	\begin{split}
	\overline{f}'(\gamma)= & -\frac{1}{\gamma}\Big[1+\sqrt{\gamma^2+1}+\frac{\gamma^2(3+\sqrt{\gamma^2+1})}{4}+\frac{\gamma^4}{16}\Big]+\\
	&\frac{\gamma^3}{4\sqrt{\gamma^2+1}}
	+\frac{\gamma(2+\sqrt{\gamma^2+1})}{2} +\frac{5\gamma^3}{16}-\Big[\ln(\frac{\gamma}{2})+C_E\Big]\times\\
	 &\Big[\frac{\gamma}{\sqrt{\gamma^2+1}}+\frac{\gamma(3+\sqrt{\gamma^2+1})}{2}
+\frac{\gamma^3}{4\sqrt{\gamma^2+1}}+\frac{\gamma^3}{4}\Big]\\
	 <&(1+\sqrt{\gamma^2+1})\Big[-\frac{1}{\gamma}+\frac{\gamma}{4}+\frac{\gamma^3}{4(\gamma^2+1)}\Big]\\
	  &-\gamma\left[\ln(\frac{\gamma}{2})+C_E\right](3+\sqrt{\gamma^2+1})\\
	<&-\frac{3(1+\sqrt{\gamma^2+1})}{5\gamma}-\gamma\left[\ln(\frac{\gamma}{2})+C_E\right](3+\sqrt{\gamma^2+1})\\
	<&0.
	\end{split}
	\end{equation*}
	Then $\overline{f}(\gamma)\geq\overline{f}(\gamma_0)>0$. The left inequality of (\ref{low1}) holds.
	
	We finally treat the right inequality of (\ref{low1}). For $\gamma\in(0, \gamma_0]$, we use (\ref{new-defi}) to have
	\begin{equation*}
	\gamma\frac{K_0(\gamma)}{K_1(\gamma)}\leq \frac{-\left[\ln(\frac{\gamma}{2})+C_E\right]\gamma^2-\frac{1}{4}\left[\ln(\frac{\gamma}{2})+C_E-1\right]\gamma^4e^{\frac{\gamma^2}{10}}}
	{1+\frac{1}{2}\left[\ln(\frac{\gamma}{2})+C_E-\frac{1}{2}\right] \gamma^2e^{\frac{5\gamma^2}{16}}}.
	\end{equation*}
	To prove the right inequality of (\ref{low1}), we only need to derive the following inequality
	\begin{align*}
	\begin{aligned}
	&-\left[\ln(\frac{\gamma}{2})+C_E\right]\gamma^2-\frac{1}{4}\left[\ln(\frac{\gamma}{2})+C_E-1\right]\gamma^4e^{\frac{\gamma^2}{10}}\leq\\
	 &  \Big[1+\frac{1}{2}\Big(\ln(\frac{\gamma}{2})+C_E-\frac{1}{2}\Big) \gamma^2 e^{\frac{5\gamma^2}{16}}\Big]\gamma^2\left[\frac{11}{16}-\left(\ln(\frac{\gamma}{2})+C_E\right)\right]
	.
	\end{aligned}
	\end{align*}
	That is,
	\begin{align}\label{001new}
	\begin{aligned}
	\tilde{f}(\gamma)=&\gamma^2\Big\{-\Big[\Big(\ln(\frac{\gamma}{2})+C_E\Big)-1\Big]e^{\frac{\gamma^2}{10}}
	+\Big[2\Big(\ln(\frac{\gamma}{2})+C_E\Big)^2-\\
	&\frac{19}{8}\Big(\ln(\frac{\gamma}{2})+C_E\Big)+\frac{11}{16}\Big]e^{\frac{5\gamma^2}{16}}\Big\}-\frac{11}{4}\leq0
	\end{aligned}
	\end{align}
	for $\gamma\in(0, \gamma_0]$. Note the fact
	\begin{align*}
	\begin{aligned}
	\tilde{f}'(\gamma)=&2\gamma\Big\{-\Big[\Big(\ln(\frac{\gamma}{2})+C_E\Big)-1\Big]e^{\frac{\gamma^2}{10}}
	+\Big[2\Big(\ln(\frac{\gamma}{2})+C_E\Big)^2-\\
&\frac{19}{8}\Big(\ln(\frac{\gamma}{2})+C_E\Big)+\frac{11}{16}\Big]e^{\frac{5\gamma^2}{16}}\Big\}-\gamma e^{\frac{\gamma^2}{10}}-\frac{19\gamma}{8}e^{\frac{5\gamma^2}{16}}+\\
	&4\Big(\ln(\frac{\gamma}{2})+C_E\Big)\gamma e^{\frac{5\gamma^2}{16}}+\frac{\gamma^3}{5}\Big[-\Big(\ln(\frac{\gamma}{2})+C_E\Big)+1\Big]e^{\frac{\gamma^2}{10}}\\
	&+\frac{5\gamma^3}{8}\Big[2\Big(\ln(\frac{\gamma}{2})+C_E\Big)^2-\frac{19}{8}\Big(\ln(\frac{\gamma}{2})+C_E\Big)
	+\frac{11}{16}\Big]e^{\frac{5\gamma^2}{16}}\\
	\geq&\gamma\Big\{\Big(\frac{\gamma^2}{5}+1\Big)e^{\frac{\gamma^2}{10}}
	+\Big[-\frac{3}{4}\Big(\ln(\frac{\gamma}{2})+C_E\Big)+\frac{55\gamma^2}{128}-1\Big]e^{\frac{5\gamma^2}{16}}\Big\}\\
>&0.
	\end{aligned}
	\end{align*}
	Here we used the estimate $-\frac{3}{4}\Big(\ln(\frac{\gamma}{2})+C_E\Big)+\frac{55\gamma^2}{128}\geq\frac{1}{2}$ for $\gamma\in(0, \gamma_0]$.
	Then we have
	\begin{align*}\label{001new}
	\begin{aligned}
	\tilde{f}(\gamma)\leq&\gamma^2\Big(e^{\frac{\gamma^2}{10}}
	+\frac{11}{16}e^{\frac{5\gamma^2}{16}}\Big)-\frac{11}{4}<0
	\end{aligned}
	\end{align*}
	for $\gamma\in(0, \gamma_0]$. (\ref{001new}) is verified.
\end{proof}

For later use, we also need two different estimates:
\begin{proposition}\label{K01p} Let $\gamma\in(\sqrt{2}, \infty)$. Then $\frac{K_0(\gamma)}{K_1(\gamma)}$ satisfies:
	\begin{equation}\label{rough}
	1-\frac{1}{2\gamma}\leq\frac{K_0(\gamma)}{K_1(\gamma)}\leq 1-\frac{1}{2\gamma}+\frac{3}{8\gamma^2}+\frac{3}{16\gamma^3}.
	\end{equation}
	Moreover, for $\gamma\in(2, \infty)$, it holds that
	\begin{align} \label{acurate}
	\begin{aligned}
	&\frac{K_0(\gamma)}{K_1(\gamma)}\geq 1-\frac{1}{2\gamma}+\frac{3}{8\gamma^2}-\frac{3}{8\gamma^3}+\frac{63}{128\gamma^4}-\frac{31}{20\gamma^5}, \\
	&\frac{K_0(\gamma)}{K_1(\gamma)}\leq 1-\frac{1}{2\gamma}+\frac{3}{8\gamma^2}-\frac{3}{8\gamma^3}+\frac{63}{128\gamma^4}+\frac{7}{8\gamma^5}.
	\end{aligned}
	\end{align}
\end{proposition}

\begin{proof} Compared to the proof of (\ref{rough}), the proof of (\ref{acurate}) is more tedious but simpler. For brevity, we only prove (\ref{rough}).
	From (\ref{coefficient}), one has
	\begin{equation}\label{eff0}
\begin{split}
	&A_{0,1}=-\frac{1}{8}, \quad A_{0,2}=\frac{9}{2\times8^2}, \quad A_{0,3}=-\frac{75}{2\times8^3}, \\
 &A_{0,4}=\frac{3\times25\times49}{8^5}, \quad A_{0,5}=-\frac{15\times49\times81}{8^6},\quad\mbox{and}
 \end{split}
 \end{equation}
 \begin{equation}\label{eff1}
 \begin{split}
	& A_{1,1}=\frac{3}{8}, \quad A_{1,2}=-\frac{15}{2\times8^2}, \quad A_{1,3}=\frac{105}{2\times8^3}, \\
 &A_{1,4}=-\frac{105\times45}{8^5}, \quad A_{1,5}=\frac{21\times45\times77}{8^6}.
 \end{split}
 \end{equation}
Moreover, for $\gamma>0$,
	\begin{eqnarray}
	&&r_{0,3}\leq2e^{-\frac{1}{4\gamma}}|A_{0,3}|= \frac{75e^{-\frac{1}{4\gamma}}}{8^3},\quad r_{1,3}\leq2e^{\frac{3}{4\gamma}}|A_{1,3}|= \frac{105e^{\frac{3}{4\gamma}}}{8^3},\label{rem3} \\
	&&r_{0,4}\leq2e^{-\frac{1}{4\gamma}}|A_{0,4}|= \frac{75\times49e^{-\frac{1}{4\gamma}}}{4\times8^5},\nonumber\\
 &&r_{1,4}\leq2e^{\frac{3}{4\gamma}}|A_{1,4}|= \frac{105\times45e^{\frac{3}{4\gamma}}}{4\times8^4}, \label{rem4} \\
	&&r_{0,5}=2e^{-\frac{1}{4\gamma}}|A_{0,5}|\leq \frac{15\times49\times81}{4\times8^5}e^{-\frac{1}{4\gamma}},\nonumber \\
	&&r_{1,5}=2e^{\frac{3}{4\gamma}}|A_{1,5}|\leq \frac{21\times45\times77}{4\times8^5}e^{\frac{3}{4\gamma}}. \nonumber
	\end{eqnarray}
		
	Firstly, we show that the inequality on the left side of (\ref{rough}) is true. We use (\ref{coefficient}), (\ref{eff0}), (\ref{eff1}) and (\ref{rem3}) to have
	$$\frac{K_0(\gamma)}{K_1(\gamma)}\geq\frac{ 1-\frac{1}{8\gamma}+\frac{9}{128\gamma^2}-\frac{75}{8^3\gamma^3}e^{-\frac{1}{4\gamma}}}{1+\frac{3}{8\gamma}-\frac{15}{128\gamma^2}+\frac{105}{8^3\gamma^3}e^{\frac{3}{4\gamma}}}.$$
	Then, it suffice to show that
	\begin{equation*}
	\begin{split}
	&1-\frac{1}{8\gamma}+\frac{9}{128\gamma^2}-\frac{75}{8^3\gamma^3}e^{-\frac{1}{4\gamma}}\geq\\
	&\left(1+\frac{3}{8\gamma}-\frac{15}{128\gamma^2}+\frac{105}{8^3\gamma^3}e^{\frac{3}{4\gamma}}\right)\left(1-\frac{1}{2\gamma}\right)=\\
	 &1-\frac{1}{8\gamma}-\left(\frac{3}{16}+\frac{15}{128}\right)\frac{1}{\gamma^2}+\\
&\left(\frac{105}{8^3}e^{\frac{3}{4\gamma}}+\frac{15}{256}\right)\frac{1}{\gamma^3}-\frac{105}{2\times8^3\gamma^4}e^{\frac{3}{4\gamma}}.
	\end{split}
	\end{equation*}
	Namely,
	\begin{align} \label{r1}
	\begin{aligned}
	 &\frac{3}{8\gamma^2}+\frac{105}{2\times8^3\gamma^4}e^{\frac{3}{4\gamma}}\geq\left(\frac{75}{8^3}e^{-\frac{1}{4\gamma}}+\frac{105}{8^3}e^{\frac{3}{4\gamma}}+\frac{15}{256}\right)\frac{1}{\gamma^3},\\
	 &\gamma^2-\left(\frac{25}{64}e^{-\frac{1}{4\gamma}}+\frac{35}{64}e^{\frac{3}{4\gamma}}+\frac{5}{32}\right)\gamma+\frac{35}{128}e^{\frac{3}{4\gamma}}\geq0.
	\end{aligned}
	\end{align}
	%Noting the simple inequality
	 %$$\frac{105}{2\times8^3\gamma^2}e^{\frac{3}{4\gamma}}+\frac{105}{2\times8^3\gamma^4}e^{\frac{3}{4\gamma}}\geq\frac{105}{8^3\gamma^3}e^{\frac{3}{4\gamma}},$$
	%we can further reduce (\ref{r1}) as
	%\begin{equation}\label{left}
	%\left(1-\frac{35}{128}e^{\frac{3}{4\gamma}}\right)\gamma-\frac{25}{64}e^{-\frac{1}{4\gamma}}-\frac{5}{32}\geq0.
	%\end{equation}
	Denote $f_3(\gamma)=:\gamma^2-\left(\frac{25}{64}e^{-\frac{1}{4\gamma}}+\frac{35}{64}e^{\frac{3}{4\gamma}}+\frac{5}{32}\right)\gamma+\frac{35}{128}e^{\frac{3}{4\gamma}}$. It is easy to check that
	$$f_3(1.1)>0,\qquad f_3'(\gamma)>0 ~~\mbox{for} ~~\gamma\geq1.$$
	Then (\ref{r1}) holds for $\gamma\in [1.1,\infty)\subset (\sqrt{2}, \infty)$.
	
	We now continue to verify the inequality on the right side of (\ref{rough}). Similarly, from (\ref{coefficient}), (\ref{eff0}), (\ref{eff1}) and (\ref{rem4}), we get
	$$\frac{K_0(\gamma)}{K_1(\gamma)}\leq\frac{ 1-\frac{1}{8\gamma}+\frac{9}{128\gamma^2}-\frac{75}{2\times8^3\gamma^3}+\frac{75\times49e^{-\frac{1}{4\gamma}}}{4\times8^4\gamma^4}}
	{1+\frac{3}{8\gamma}-\frac{15}{128\gamma^2}+\frac{105}{2\times8^3\gamma^3}-\frac{105\times45e^{\frac{3}{4\gamma}}}{4\times8^4\gamma^4}}.$$
	The proof can be completed if we can show the following inequality:
	\begin{equation*}
	\begin{split}
	&1-\frac{1}{8\gamma}+\frac{9}{128\gamma^2}-\frac{75}{2\times8^3\gamma^3}+\frac{75\times49e^{-\frac{1}{4\gamma}}}{4\times8^4\gamma^4}\leq\\
	 &\Big(1+\frac{3}{8\gamma}-\frac{15}{128\gamma^2}+\frac{105}{2\times8^3\gamma^3}-\frac{105\times45e^{\frac{3}{4\gamma}}}{4\times8^4\gamma^4}\Big)\times\\
	&\Big(1-\frac{1}{2\gamma}+\frac{3}{8\gamma^2}+\frac{3}{16\gamma^3}\Big)=\\
	&1-\frac{1}{8\gamma}+\frac{9}{128\gamma^2}+\left(\frac{3}{16}+\frac{9}{64}+\frac{15}{256}+\frac{105}{1024}\right)\frac{1}{\gamma^3}\\
	&+\left(\frac{9}{128}-\frac{45}{1024}-\frac{105}{4\times8^3}-\frac{105\times45}{4\times8^4}e^{\frac{3}{4\gamma}}\right)\frac{1}{\gamma^4}\\
	&+\left(\frac{-45\times4+315}{2\times8^4}+\frac{105\times45}{8^5}e^{\frac{3}{4\gamma}}\right)\frac{1}{\gamma^5}\\
	 &+\left(\frac{315}{4\times8^4}-\frac{315\times45}{4\times8^5}e^{\frac{3}{4\gamma}}\right)\frac{1}{\gamma^6}-\frac{315\times45}{8^6}\frac{e^{\frac{3}{4\gamma}}}{\gamma^7}.
	\end{split}
	\end{equation*}
	This inequality can be simplified as
	\begin{equation}\label{right}
	\begin{split}
	&9\gamma^4+\left(\frac98-\frac{195}{128}-\frac{105\times45e^{\frac{3}{4\gamma}}+75\times49e^{-\frac{1}{4\gamma}}}{2\times8^3}\right)\gamma^3+\\
	&\hspace{1cm}
	 \Big(-\frac{45}{128}+\frac{315}{8^3}+\frac{105\times45e^{\frac{3}{4\gamma}}}{4\times8^3}\Big)\gamma^2
+\\
	&\hspace{1cm}\Big(\frac{315}{2\times8^3}-\frac{315\times45e^{\frac{3}{4\gamma}}}{2\times8^4}\Big)\gamma-\frac{315\times45e^{\frac{3}{4\gamma}}}{2\times8^4}\geq0.
	\end{split}
	\end{equation}
	Denote the function on the left side of (\ref{right}) as $f_4(\gamma)$. It can be verified that
	\begin{equation*}
	\begin{split}
	f_4(\sqrt{2})>0,\qquad f'_4(\gamma)>0 ~~\mbox{for}~~ \gamma\in[1,\infty) .
	\end{split}
	\end{equation*}
	Therefore, (\ref{right}) holds for $\gamma\in [\frac{5}{4}, \infty)$
	
	%since
	%\begin{eqnarray*}
	%&&\frac{1}{2\gamma^2}+\left(\frac{7}{40}-\frac{105}{2\times8^3}\right)\frac{1}{\gamma^4}\\
	%&\geq&\frac{13}{20\gamma^2}+\left(\frac{7}{40}-\frac{123}{2\times8^3}\right)\Big(\frac{1}{\gamma^2}+\frac{1}{\gamma^4}\Big)\\
	%&\geq&\left(\frac{7}{20}-\frac{123}{8^3}\right)\frac{1}{\gamma^3}+\frac{49}{160\gamma^5}.
	%\end{eqnarray*}
	
\end{proof}

%%%%%%%%%%%%%%%%%%%%%%%%%%%%%%%%%
%
%
%%%%%%%%%%%%%%%%%%%%%%%%%%%%%%%%%
\section*{Appendix 3: Essential estimates and solution of the conjectures in \cite{Speck-Strain-CMP-2011}}
In this part, we present estimates essential to the analysis in the rest of our paper. The estimates are also closely related to the two conjectures in \cite{Speck-Strain-CMP-2011}. For convenience of discussion, we first list these conjectures.

\begin{conjecture}\label{conj1}
The first  conjecture of \cite{Speck-Strain-CMP-2011} reads:
The map $(n,\gamma)\rightarrow(\mathfrak{H}(n,\gamma),\mathfrak{P}(n,\gamma))$ is auto-diffeomorphism of the region $(0,\infty)\times(0,\infty)$, where the maps $\mathfrak{H}$ and $\mathfrak{P}$ are defined as follows:
	\begin{eqnarray*}
		&&S=\mathfrak{H}(n,\gamma)=k_B\ln\left(\frac{4\pi e^4m^3c^2h^{-3}K_2(\gamma)}{n\gamma}e^{\gamma\frac{K_1(\gamma)}{K_2(\gamma)}}\right),\\
		&&p=\mathfrak{P}(n,\gamma)=\frac{nmc^2}{\gamma}.
	\end{eqnarray*}
\end{conjecture}
With relations in (\ref{press})-(\ref{num-den}), one can deduce the local resolvability of any one of the variables $n, T, S$ and $p$ in terms of any two of the others whenever one knows that the necessary derivatives are non-zero. In fact, as in the analysis of Lemma 3.5 in \cite{Speck-Strain-CMP-2011}, the negativity of $\frac{\partial p}{\partial \gamma}\Big|_{S}$:
\begin{equation}\label{inver}
\frac{\partial_\gamma|_{S}p}{p}=\gamma\left(\frac{K_1(\gamma)}{K_2(\gamma)}\right)^2+3\frac{K_1(\gamma)}{K_2(\gamma)}-\gamma-\frac{4}{\gamma}<0,\quad \gamma>0
\end{equation}
would imply Conjecture \ref{conj1}. Here and in the rest part of this paper, we use the notation
$$\partial_X|_{Y}$$
to denote partial differentiation with respect to the variable $X$ while $Y$ is held constant.

Authors in \cite{Speck-Strain-CMP-2011} also made another conjecture which is about the speed of sound in the relativistic setting, a stronger statement than Conjecture \ref{conj1}:
\begin{conjecture} \label{conj2}
The second  conjecture of \cite{Speck-Strain-CMP-2011} reads:
 Under the relations (\ref{press})-(\ref{num-den}), p can be written as a smooth, positive
	function of $E, S$ on the domain $(0,\infty) \times (0,\infty)$, i.e., the kinetic equation of state:
	$$p=p(e, S)$$
	is well-defined for all $(e, S) \in (0,\infty)\times(0,\infty)$. Furthermore, on $(0,\infty)\times(0,\infty)$, we
	have that
	\begin{equation}\label{sou-spe}
	0<\frac{\partial p}{\partial e}\Big|_{S}(e, S)=\frac{\partial p(e, S)}{\partial e}\Big|_{S}<\frac{1}{3}.
	\end{equation}
\end{conjecture}

As is noted in \cite{Speck-Strain-CMP-2011}, proving (\ref{sou-spe}) is equivalent to proving the following inequality
\begin{equation}\label{speed}
3<\frac{\partial e}{\partial p}\Big|_{S}(p, S)=3+\gamma\frac{K_1(\gamma)}{K_2(\gamma)}+\frac{\gamma\left(\frac{K_1(\gamma)}{K_2(\gamma)}\right)^2
+4\frac{K_1(\gamma)}{K_2(\gamma)}-\gamma}{\gamma\left(\frac{K_1(\gamma)}{K_2(\gamma)}\right)^2
+3\frac{K_1(\gamma)}{K_2(\gamma)}-\gamma-\frac{4}{\gamma}}<\infty,
\end{equation}
since $\frac{\partial p}{\partial e}\Big|_{S}(e, S)=\left(\frac{\partial e}{\partial p}\Big|_{S}(p, S)\right)^{-1}$. The speed of sound, the square of which is defined to be $\frac{\partial p}{\partial e}\Big|_{S}(e, S)$, is a fundamental quantity in the relativistic Euler system. $0<\frac{\partial p}{\partial e}\Big|_{S}(e, S)<1$ is a fundamental thermodynamic assumption for physically relevant equations of state. Moreover, in  \cite{Speck-Strain-CMP-2011}, the non-negativity of $\frac{\partial p}{\partial e}\Big|_{S}(e, S)$ plays a fundamental role in the well-posedness theory of the relativistic Euler system,  see Remark 2.1 there.

\begin{remark} For $\gamma\in(0,\frac{1}{10}]\cup[70,\infty)$, the authors in \cite{Speck-Strain-CMP-2011} verified (\ref{inver}) and made Conjectures \ref{conj1} and \ref{conj2} based on numerical observations in the remained region $\gamma\in(\frac{1}{10},70)$. Later, Juan \cite{Juan-2013} gave a proof of the two conjectures for any range of $\gamma$. We will present two estimates which implies the two conjectures since these estimates are basic in our analysis.
\end{remark}

In the following proposition, based on Lemma \ref{def-pro} and Corollary \ref{K012}, we present estimates more accurate than (\ref{inver}) and (\ref{speed}).
\begin{proposition}\label{eep} Let $\gamma\in(0,\infty)$ and $K_j(\gamma) (j\geq0)$ be the functions defined in Lemma \ref{def-pro}. Then it holds that
	\begin{equation}\label{imp-ine1}
	\gamma^2\left(\frac{K_1(\gamma)}{K_2(\gamma)}\right)^2+3\gamma\frac{K_1(\gamma)}{K_2(\gamma)}- \gamma^2-3<0,
	\end{equation}
	\begin{equation}\label{imp-ine2}
	\gamma\left(\frac{K_1(\gamma)}{K_2(\gamma)}\right)^3+4\left(\frac{K_1(\gamma)}{K_2(\gamma)}\right)^2-\gamma\frac{K_1(\gamma)}{K_2(\gamma)}-1<0.
	\end{equation}
	%Moreover, the range $R_f$ of $f(\gamma)$ is
%	\begin{equation}\label{range}
%	R_f=\left\{f(\gamma)|\gamma\in(0,\infty)\right\}=(3, \infty).
%	\end{equation}
\end{proposition}
\begin{proof} We first prove (\ref{imp-ine1}). Its proof is divided into two cases: $\gamma\in(0,\sqrt{2}]$ and $\gamma\in(\sqrt{2},\infty)$. Firstly, by (\ref{transform}), we can rewrite (\ref{imp-ine1}) as
	\begin{equation}\label{imp-ine11}
	 (\gamma^2+3)\left(\frac{K_0(\gamma)}{K_1(\gamma)}\right)^2+\left(\gamma+\frac{12}{\gamma}\right)\frac{K_0(\gamma)}{K_1(\gamma)}+\frac{12}{\gamma^2}-\gamma^2-2>0.
	\end{equation}
	Noting $K_0(\gamma), K_1(\gamma)>0$ for $\gamma\in(0,\infty)$ and
	$$\frac{12}{\gamma^2}-\gamma^2-2>0,\quad \gamma\in(0,\sqrt{2}],$$
	(\ref{imp-ine1}) holds when $\gamma\in(0,\sqrt{2}]$. For the case $\gamma\in(\sqrt{2},\infty)$, we use (\ref{rough}) to have
	\begin{align*}
	\begin{aligned}
	 &(\gamma^2+3)\left(\frac{K_0(\gamma)}{K_1(\gamma)}\right)^2+\left(\gamma+\frac{12}{\gamma}\right)\frac{K_0(\gamma)}{K_1(\gamma)}+\frac{12}{\gamma^2}-\gamma^2-2 \geq\nonumber\\
	 &(\gamma^2+3)\left(1-\frac{1}{2\gamma}\right)^2+\left(\gamma+\frac{12}{\gamma}\right)\left(1-\frac{1}{2\gamma}\right)+\frac{12}{\gamma^2}-\gamma^2-2=\\
	&\gamma^2-\gamma+\frac{13}{4}-\frac{3}{\gamma}+\frac{3}{4\gamma^2}+\gamma-\frac{1}{2}+\frac{12}{\gamma}-\frac{6}{\gamma^2}-\gamma^2-2= \\
	&\frac{3}{4}+\frac{9}{\gamma}-\frac{21}{4\gamma^2}>0.
	\end{aligned}
	\end{align*}
	This yields (\ref{imp-ine11}). Then (\ref{imp-ine1}) follows.
	
	Now we turn to prove  (\ref{imp-ine2}). The proof is also done in two cases, $\gamma\in(0,\sqrt{2}]$ and $\gamma\in(\sqrt{2},\infty)$, separately. We use (\ref{transform}) again to rewrite (\ref{imp-ine2}) as
	\begin{equation}\label{imp-lin21}
	 \left(\frac{K_0(\gamma)}{K_1(\gamma)}\right)^3+\left(\gamma+\frac{6}{\gamma}\right)\left(\frac{K_0(\gamma)}{K_1(\gamma)}\right)^2+\frac{12}{\gamma^2}\frac{K_0(\gamma)}{K_1(\gamma)}-\gamma-\frac{4}{\gamma}+\frac{8}{\gamma^3}>0.
	\end{equation}
	We first show that (\ref{imp-lin21}) is true when $\gamma\in(0,\sqrt{2})$. For this purpose, we use (\ref{K012-2}) to have
	\begin{align*}
		 &\left(\frac{K_0(\gamma)}{K_1(\gamma)}\right)^3+\left(\gamma+\frac{6}{\gamma}\right)
\left(\frac{K_0(\gamma)}{K_1(\gamma)}\right)^2+\frac{12}{\gamma^2}\frac{K_0(\gamma)}{K_1(\gamma)}-\gamma-\frac{4}{\gamma}+\frac{8}{\gamma^3} =\nonumber\\	 &\frac{K_0(\gamma)}{K_1(\gamma)}\left[\left(\frac{K_0(\gamma)}{K_1(\gamma)}\right)^2+\frac{2}{\gamma}\frac{K_0(\gamma)}{K_1(\gamma)}-\frac{1}{3}\right]
	 +\\
&\left(\gamma+\frac{4}{\gamma}\right)\left[\left(\frac{K_0(\gamma)}{K_1(\gamma)}\right)^2
+\frac{2}{\gamma}\frac{K_0(\gamma)}{K_1(\gamma)}-\frac{1}{3}\right]+\\
	 &\frac{1}{3}\left(\gamma+\frac{4}{\gamma}\right)+\Big(\frac{1}{3}-2+\frac{4}{\gamma^2}\Big)
\frac{K_0(\gamma)}{K_1(\gamma)}-\gamma-\frac{4}{\gamma}+\frac{8}{\gamma^3}>  \\
	&\left(\frac{4}{\gamma^2}-\frac{5}{3}\right)\frac{K_0(\gamma)}{K_1(\gamma)}-\frac{2\gamma}{3}-\frac{8}{3\gamma}+\frac{8}{\gamma^3}>0,
		\end{align*}
	when $\gamma\in(0,\sqrt{2}]$. Here we have used the simple estimates: for $\gamma\in(0,\sqrt{2}]$,
	$$\frac{4}{\gamma^2}-\frac{5}{3}>0,\quad -\frac{2\gamma}{3}-\frac{8}{3\gamma}+\frac{8}{\gamma^3}\geq0.$$
	When $\gamma\in(\sqrt{2},\infty)$, similar to proof of (\ref{imp-ine1}), we use (\ref{rough}) to obtain
	\begin{align}
	\begin{aligned}
	 &\Big(\frac{K_0(\gamma)}{K_1(\gamma)}\Big)^3+\left(\gamma+\frac{6}{\gamma}\right)\left(\frac{K_0(\gamma)}{K_1(\gamma)}\right)^2
+\frac{12}{\gamma^2}\frac{K_0(\gamma)}{K_1(\gamma)}-\gamma-\frac{4}{\gamma}+\frac{8}{\gamma^3}\geq\nonumber\\
	&\left(\gamma+1+\frac{11}{2\gamma}\right)\left(1-\frac{1}{\gamma}+\frac{1}{4\gamma^2}\right)+
	\frac{12}{\gamma^2}+\frac{2}{\gamma^3}-\gamma-\frac{4}{\gamma}=\\
	&\gamma+\frac{19}{4\gamma}-\frac{21}{4\gamma^2}+\frac{11}{8\gamma^3}+
	\frac{12}{\gamma^2}+\frac{2}{\gamma^3}-\gamma-\frac{4}{\gamma} =\\
	&\frac{3}{4\gamma}+\frac{27}{4\gamma^2}+\frac{27}{8\gamma^3}>0.
	%\\=&\frac{1}{8\gamma^3}\left(6\gamma^2+54\gamma+27\right)>0.
	\end{aligned}
	\end{align}
\end{proof}
\begin{remark} In the proof of Proposition \ref{eep}, instead of working on the ratio $\frac{K_1(\gamma)}{K_2(\gamma)}$ directly, we transformed the ratio $\frac{K_1(\gamma)}{K_2(\gamma)}$ into the ratio $\frac{K_0(\gamma)}{K_1(\gamma)}$ and divided our proof of (\ref{imp-ine1}) and (\ref{imp-ine2}) into two cases,  $\gamma\in(0,\sqrt{2}]$ and $\gamma\in(\sqrt{2},\infty)$. The motivations for this are as follows: from the expansion of $K_j(\gamma)$ in (\ref{remainder})  and (\ref{coefficient}), we can see that it works well for $\gamma$ which is a little larger than $1$, and vice versa; the estimate of remained term $|r_{j,n}(\gamma)|$ seems more accurate when $j$ is smaller due to the coefficient $e^{[j^2-1/4]\gamma^{-1}}$ in the estimate, which increase more rapidly than the normal exponential function; when $\gamma$ is small, we can make use of the simple inequality (\ref{K012-2}) from the observation (\ref{K012-1}).
	
\end{remark}

\begin{remark}
	Estimates (\ref{imp-ine1}) and (\ref{imp-ine2}) are more accurate than (\ref{inver}) and (\ref{sou-spe}).  Then Conjecture \ref{conj1} and \ref{conj2} are correct and the main results in \cite{Speck-Strain-CMP-2011}
	can be extended to including the whole case $\gamma\in(0,\infty)$. The range for speed of sound $\sqrt{\frac{\partial p}{\partial e}\Big|_{S}(E, S)}$ is $(0,\frac{\sqrt{3}}{3})$ for $\gamma\in(0,\infty)$. Moreover, estimates (\ref{imp-ine1}) and (\ref{imp-ine2}) are of essential importance in this paper.
\end{remark}

\begin{proposition}\label{peep} Let $\gamma\in(0,\infty)$ and $K_j(\gamma) (j\geq0)$ be the functions defined in Lemma \ref{def-pro}. Then it holds that
	\begin{equation}\label{imp-inep}
	\gamma^2\left(\frac{K_0(\gamma)}{K_1(\gamma)}\right)^3+2\gamma\left(\frac{K_0(\gamma)}{K_1(\gamma)}\right)^2-(\gamma^2
	+2)\frac{K_0(\gamma)}{K_1(\gamma)}- \gamma<0.
	\end{equation}
\end{proposition}
\begin{proof}
	For $\gamma\leq2,$ it is straightforward to get (\ref{imp-inep}) by the fact $\frac{K_0(\gamma)}{K_1(\gamma)}<1$. For the case $\gamma>2$, we use (\ref{rough}) to have
	$$\frac{K_0(\gamma)}{K_1(\gamma)}\leq 1-\frac{1}{2\gamma}+\frac{1}{2\gamma^2},$$
	and
	\begin{align*}
	\begin{aligned}
	&\gamma^2\left(\frac{K_0(\gamma)}{K_1(\gamma)}\right)^2+2\gamma\left(\frac{K_0(\gamma)}{K_1(\gamma)}\right)^2-(\gamma^2
	+2)\frac{K_0(\gamma)}{K_1(\gamma)}- \gamma\leq\\
	& \gamma\frac{K_0(\gamma)}{K_1(\gamma)}\Big[\gamma\Big(1-\frac{1}{2\gamma}
	+\frac{1}{2\gamma^2}\Big)^2+2\Big(1-\frac{1}{2\gamma}+\frac{1}{2\gamma^2}\Big)\Big]-\\
&(\gamma^2
	+2)\frac{K_1(\gamma)}{K_2(\gamma)}- \gamma\leq\\
	&\gamma\frac{K_0(\gamma)}{K_1(\gamma)}\Big[\gamma\Big(1-\frac{1}{\gamma}
	+\frac{5}{4\gamma^2}-\frac{1}{2\gamma^3}+\frac{1}{4\gamma^4}\Big)+2-\frac{1}{\gamma}+\frac{1}{\gamma^2}\Big)\Big]-\\
&(\gamma^2
	+2)\frac{K_0(\gamma)}{K_1(\gamma)}- \gamma\leq\\
	&\frac{K_0(\gamma)}{K_1(\gamma)}\Big(\gamma-\frac{3}{4}+\frac{1}{2\gamma}+\frac{1}{4\gamma^2}\Big)-\gamma<0
	.
	\end{aligned}
	\end{align*}
\end{proof}
%%%%%%%%%%%%%%%%%%%%%
%
%
%%%%%%%%%%%%%%%%%%%%%%%

%%%%%%%%%%%%%%%%%%%%%%%%%%%%%%%%%
%

%%%%%%%%%%%%%%%%%%%%%%%%%%%%%%%%%
%
%
%%%%%%%%%%%%%%%%%%%%%%%%%%%%%%%%%
\section*{Appendix 4: Proof of Proposition \ref{genuine} for the genuine nonlinearity}
$\newline$
{\it{Proof of Proposition \ref{genuine}:}}   We only need to prove (\ref{negativity}). If this is done, (\ref{genuine0}) follows immediately from (\ref{genui}) and (\ref{negativity}).
	From (\ref{pz}), we can further obtain
	\begin{align}\label{pp}
		e_{pp}=&\frac{e_{p}}{p}-\frac{e}{p^2}+\frac{1}{\partial_\gamma p}\partial_\gamma\left(\frac{p}{\partial_\gamma p}\frac{d}{d\gamma}\left(\gamma\frac{K_1(\gamma)}{K_2(\gamma)}\right)\right)=\nonumber\\
	 &\frac{1}{p}\frac{\gamma\left(\frac{K_1(\gamma)}{K_2(\gamma)}\right)^2
+4\frac{K_1(\gamma)}{K_2(\gamma)}-\gamma}{\gamma\left(\frac{K_1(\gamma)}{K_2(\gamma)}\right)^2+3\frac{K_1(\gamma)}{K_2(\gamma)}
-\gamma-\frac{4}{\gamma}}	+\nonumber\\
	&\frac{1}{p}\frac{1}{\gamma\left(\frac{K_1(\gamma)}{K_2(\gamma)}\right)^2+3\frac{K_1(\gamma)}{K_2(\gamma)}-\gamma-\frac{4}{\gamma}}\times\\
&	\frac{d}{d\gamma}\Big(\frac{\gamma\Big(\frac{K_1(\gamma)}{K_2(\gamma)}\Big)^2
		 +4\frac{K_1(\gamma)}{K_2(\gamma)}-\gamma}{\gamma\Big(\frac{K_1(\gamma)}{K_2(\gamma)}\Big)^2+3\frac{K_1(\gamma)}{K_2(\gamma)}
-\gamma-\frac{4}{\gamma}}\Big).\nonumber
		\end{align}
	Then we use (\ref{pz}) and (\ref{pp}) to have
	\begin{align}\label{pp1}
	\begin{aligned}
	&(e+p)e_{pp}-2e_{p}(e_{p}-1)=\\
	&e_{p}(-2e_{p}+3)+\frac{e}{p}\Big(e_{p}-\frac{e}{p}-1\Big)
	+\frac{e+p}{\partial_\gamma p}\partial_\gamma\Big(\frac{p}{\partial_\gamma p}\frac{d}{d\gamma}\Big(\gamma\frac{K_1(\gamma)}{K_2(\gamma)}\Big)\Big)<
	\\
	 &-9+\Big(\gamma\frac{K_1(\gamma)}{K_2(\gamma)}+3\Big)\frac{\frac{K_1(\gamma)}{K_2(\gamma)}+\frac{4}{\gamma}}
{\gamma\Big(\frac{K_1(\gamma)}{K_2(\gamma)}\Big)^2+3\frac{K_1(\gamma)}{K_2(\gamma)}-\gamma-\frac{4}{\gamma}}+\\
	 &\Big(\gamma\frac{K_1(\gamma)}{K_2(\gamma)}+4\Big)\Bigg[\frac{\frac{1}{\gamma}}{\gamma\Big(\frac{K_1(\gamma)}{K_2(\gamma)}\Big)^2
+3\frac{K_1(\gamma)}{K_2(\gamma)}-\gamma-\frac{4}{\gamma}}-\Big(\frac{K_1(\gamma)}{K_2(\gamma)}+\frac{4}{\gamma}\Big)\times\\
	 &\frac{\Big[2\gamma\Big(\frac{K_1(\gamma)}{K_2(\gamma)}\Big)^3+10\Big(\frac{K_1(\gamma)}{K_2(\gamma)}\Big)^2
		+\Big(\frac{9}{\gamma}-2\gamma\Big)\frac{K_1(\gamma)}{K_2(\gamma)}-4+\frac{4}{\gamma^2}\Big]}
	{\Big(\gamma\Big(\frac{K_1(\gamma)}{K_2(\gamma)}\Big)^2+3\frac{K_1(\gamma)}{K_2(\gamma)}-\gamma-\frac{4}{\gamma}\Big)^3}\Bigg]=\\
	&-9+\frac{\left(\gamma\frac{K_1(\gamma)}{K_2(\gamma)}+4\right)\left(\frac{K_1(\gamma)}{K_2(\gamma)}+\frac{4}{\gamma}\right)}
	{\left(\gamma\left(\frac{K_1(\gamma)}{K_2(\gamma)}\right)^2+3\frac{K_1(\gamma)}{K_2(\gamma)}-\gamma-\frac{4}{\gamma}\right)^3}\times\mathcal {I}_1(\gamma)
	\Big],
	\end{aligned}
	\end{align}
	where
	\begin{align}
	\begin{aligned}
	\mathcal {I}_1(\gamma)=&\gamma^2\Big(\frac{K_1(\gamma)}{K_2(\gamma)}\Big)^4+4\gamma\Big(\frac{K_1(\gamma)}{K_2(\gamma)}\Big)^3
-(2\gamma^2+9)\left(\frac{K_1(\gamma)}{K_2(\gamma)}\right)^2-\nonumber\\
	&\Big(4\gamma+\frac{33}{\gamma}\Big)\frac{K_1(\gamma)}{K_2(\gamma)}+\gamma^2+12+\frac{12}{\gamma^2}.
	\end{aligned}
	\end{align}
	
	%\begin{align}
	%\begin{aligned}
	%\mathcal {I}_1(\gamma)=&\Big(\gamma^2+12+\frac{12}{\gamma^2}\Big)\left(\frac{K_0(\gamma)}{K_1(\gamma)}\right)^4+\Big(4\gamma+\frac{63}{\gamma}+\frac{96}{\gamma^3}\Big)\left(\frac{K_0(\gamma)}{K_1(\gamma)}\right)^3\\
	%&+\Big(-2\gamma^2-9+\frac{90}{\gamma^2}+\frac{288}{\gamma^4}\Big)\left(\frac{K_0(\gamma)}{K_1(\gamma)}\right)^2\nonumber\\
	 %&+\left(-4\gamma-\frac{52}{\gamma}-\frac{12}{\gamma^3}+\frac{384}{\gamma^5}\right)\frac{K_0(\gamma)}{K_1(\gamma)}+\gamma^2-\frac{52}{\gamma^2}-\frac{72}{\gamma^4}+\frac{192}{\gamma^6}>0.
	%\end{aligned}
	%\end{align}
	Noting
	$$\gamma\left(\frac{K_1(\gamma)}{K_2(\gamma)}\right)^2+3\frac{K_1(\gamma)}{K_2(\gamma)}-\gamma-\frac{4}{\gamma}<0,$$
	in order to show (\ref{negativity}), one suffices to prove
	\begin{equation}\label{p10}
	\begin{split}
	\mathcal {I}_1(\gamma)>0,
	\end{split}
	\end{equation}
	in (\ref{pp1}). By using $K_2(\gamma)=\frac{2}{\gamma}K_1(\gamma)+K_0(\gamma)$, we can rewrite (\ref{p10}) as
	\begin{align}\label{p01}
	\begin{aligned}
	\mathcal {I}_2(\gamma)=&\Big(\gamma^2+12+\frac{12}{\gamma^2}\Big)\Big(\frac{K_0(\gamma)}{K_1(\gamma)}\Big)^4
+\Big(4\gamma+\frac{63}{\gamma}+\frac{96}{\gamma^3}\Big)\Big(\frac{K_0(\gamma)}{K_1(\gamma)}\Big)^3+
	\\
	&\Big(-2\gamma^2-9+\frac{90}{\gamma^2}+\frac{288}{\gamma^4}\Big)\Big(\frac{K_0(\gamma)}{K_1(\gamma)}\Big)^2+\\
	 &+\Big(-4\gamma-\frac{52}{\gamma}-\frac{12}{\gamma^3}+\frac{384}{\gamma^5}\Big)\frac{K_0(\gamma)}{K_1(\gamma)}+\\
&\gamma^2-\frac{52}{\gamma^2}-\frac{72}{\gamma^4}+\frac{192}{\gamma^6}>0.
	\end{aligned}
	\end{align}
	Now we come to prove (\ref{p01}). It is easy to find that (\ref{p01}) holds for $\gamma\in (0, r_0]$.
	
	Now  we turn to show that (\ref{p01}) holds for $\gamma\in (r_0, \infty)$. Rewrite $\mathcal {I}_2(\gamma)$ as
	\begin{align*}
	\begin{aligned}
	\mathcal {I}_2(\gamma)=&\Big(\frac{K_0(\gamma)}{K_1(\gamma)}-1+\frac{1}{2\gamma}\Big)\Big[\Big(\gamma^2+12+\frac{12}{\gamma^2}\Big)\Big(\frac{K_0(\gamma)}{K_1(\gamma)}\Big)^3\\
	&+\Big(\gamma^2+\frac{7\gamma}{2}+
	12+\frac{57}{\gamma}+\frac{12}{\gamma^2}+\frac{90}{\gamma^3}\Big)\Big(\frac{K_0(\gamma)}{K_1(\gamma)}\Big)^2\\
	 &+\Big(-\gamma^2+3\gamma+\frac{5}{4}+\frac{51}{\gamma}+\frac{147}{2\gamma^2}+{84}{\gamma^3}+\frac{243}{\gamma^4}\Big)\frac{K_0(\gamma)}{K_1(\gamma)}\\
	&
	 -\gamma^2-\frac{\gamma}{2}-\frac{1}{4}-\frac{13}{8\gamma}+\frac{48}{\gamma^2}+\frac{141}{4\gamma^3}+\frac{201}{\gamma^4}+\frac{525}{2\gamma^5}\Big]\\
	&-\frac{3}{2\gamma}-\frac{51}{16\gamma^2}+\frac{45}{4\gamma^3}+\frac{891}{8\gamma^4}+\frac{162}{\gamma^5}+\frac{243}{4\gamma^6}=\\
	 &\Big(\frac{K_0(\gamma)}{K_1(\gamma)}-1+\frac{1}{2\gamma}\Big)\Big\{\Big(\frac{K_0(\gamma)}{K_1(\gamma)}-1+\frac{1}{2\gamma}\Big)\times\\
&\Big[\Big(\gamma^2+12+\frac{12}{\gamma^2}\Big)
	\Big(\frac{K_0(\gamma)}{K_1(\gamma)}\Big)^2+\\
	&\Big(2\gamma^2+3\gamma+24+\frac{51}{\gamma}+\frac{24}{\gamma^2}+\frac{84}{\gamma^3}\Big)\frac{K_0(\gamma)}{K_1(\gamma)}+\\
	&\gamma^2+5\gamma+\frac{95}{4}+\frac{90}{\gamma}+\frac{72}{\gamma^2}+\frac{156}{\gamma^3}+\frac{201}{\gamma^4}\Big]+\\
	&4\gamma+21+\frac{153}{2\gamma}+\frac{75}{\gamma^2}+\frac{621}{4\gamma^3}+\frac{324}{\gamma^4}+\frac{162}{\gamma^5}\Big\}-\\
	&\frac{3}{2\gamma}-\frac{51}{16\gamma^2}+\frac{45}{4\gamma^3}+\frac{891}{8\gamma^4}+\frac{162}{\gamma^5}+\frac{243}{4\gamma^6}
	\end{aligned}
	\end{align*}
	Note that
	$$-\frac{3}{2\gamma}-\frac{51}{16\gamma^2}+\frac{45}{4\gamma^3}+\frac{881}{8\gamma^4}+\frac{162}{\gamma^5}+\frac{243}{4\gamma^6}>0, \quad\mbox{for}~~\gamma\leq4.$$
	Then (\ref{p01}) holds for $\gamma\in (0, 4]$.
	
 Finally we show (\ref{p01}) for $\gamma>4$. For the case $\gamma>4$, we use (\ref{acurate}) to have
	$$\frac{K_0(\gamma)}{K_1(\gamma)}\geq 1-\frac{1}{2\gamma}+\frac{3}{8\gamma^2}-\frac{3}{8\gamma^3}.$$
	Moreover, we have
	\begin{align*}
	\begin{aligned}
	&\Big(4\gamma+21+\frac{153}{2\gamma}\Big)\Big(\frac{3}{8\gamma^2}-\frac{3}{8\gamma^3}\Big)
	-\frac{3}{2\gamma}-\frac{51}{16\gamma^2}+\\
&\frac{45}{4\gamma^3}+\frac{891}{8\gamma^4}+\frac{162}{\gamma^5}+\frac{243}{4\gamma^6}=\\
	&\Big(\frac{513}{16}+\frac{45}{4}\Big)\frac{1}{\gamma^3}
	+\frac{1323}{16\gamma^4}+\frac{162}{\gamma^5}+\frac{243}{4\gamma^6}>0
	\end{aligned}
	\end{align*}
Then we prove (\ref{p01}) for $\gamma>4$.
\bigskip

%%%%%%%%%%%%%%%%%%%%%
%
%
%%%%%%%%%%%%%%%%%%%%%%%

{\small
\textbf{Acknowledgments:} The work of T. Ruggeri was supported by GNFM (INdAM),
 the work of Q. H. Xiao was supported by grants from Youth Innovation Promotion Association and the National Natural Science Foundation of China under contract 11871469  and the work of H. J. Zhao was supported in part by the grants from National Natural Science Foundation of China under contracts 11671309 and 11731008.}

\end{document}